\documentclass[11pt,a4paper,reqno]{amsart}%
\usepackage{amsthm,amsmath,amsfonts,amssymb,amsxtra,appendix,bookmark,dsfont,bm,color}
\usepackage[shortlabels]{enumitem}

\usepackage{xcolor}
\definecolor{light-gray}{gray}{0.70}
\definecolor{dark-gray}{gray}{0.40}
\definecolor{very-light-gray}{gray}{0.90}
\usepackage{tikz}
\usepackage{tikz-3dplot}
\usetikzlibrary{3d,shadings}
\usetikzlibrary{external}
\theoremstyle{plain}
\newtheorem{theorem}{Theorem}
\newtheorem{lemma}[theorem]{Lemma}

\newtheorem{remark}[theorem]{Remark}

\theoremstyle{definition}
\newtheorem*{definition}{Definition}

\newcounter{step} 
\newcommand{\proofstep}{\par\refstepcounter{step} \noindent {\bf Step~\thestep.\space }\ignorespaces}
\usepackage{etoolbox}
\AtBeginEnvironment{proof}{\setcounter{step}{0}}

\DeclareMathOperator*{\esssup}{ess\,sup}

\def\bq{\begin{eqnarray}}
\def\eq{\end{eqnarray}}
\def\bqq{\begin{align*}}
\def\eqq{\end{align*}}

\def\nn{\nonumber}
\def\minus {\backslash}
\def\eps{\varepsilon}
\def\wto{\rightharpoonup}
\newcommand{\norm}[1]{\left\lVert #1 \right\rVert}

\renewcommand{\Re}{\operatorname{Re}}

\newcommand\1{{\ensuremath {\mathds 1} }}

\newcommand*\dotv{{}\cdot{}}

\def\R {\mathbb{R}}

\def\N {\mathcal{N}}

\def\cE {\mathcal{E}}

\def\R {\mathbb{R}}

\def\N {\mathbb{N}}

\def\sS{\mathbb{S}}
\def\d{{\, \rm d}}

\newcommand{\abs}[1]{\lvert#1\rvert}

\newcommand{\blue}{\ignorespaces}

\title[Nonexistence in Thomas-Fermi-Dirac-von Weizs\"acker theory] {Nonexistence in Thomas-Fermi-Dirac-von Weizs\"acker theory with small nuclear charges}

\author[P.T. Nam]{Phan Th\`anh Nam}
\address{Institute of Science and Technology Austria, Am Campus 1, 3400 Klosterneuburg, Austria} 
\email{pnam@ist.ac.at}

\author[H. Van Den Bosch]{Hanne Van Den Bosch}
\address{Instituto de F\'{i}sica, Pontificia Universidad Cat\'olica de Chile, Av. Vicu\~na Mackenna 4860, Santiago, Chile} 
\email{hannevdbosch@fis.puc.cl}

\begin{document}

\begin{abstract} 
We study the ionization problem in the Thomas-Fermi-Dirac-von Weizs\"acker theory 
for atoms and molecules. We prove the nonexistence of minimizers for the energy functional when the number of electrons is large and the total nuclear charge is small. This nonexistence result also applies to external potentials decaying faster {\blue than} the Coulomb potential. In the case of arbitrary nuclear charges, we obtain the nonexistence of stable minimizers and radial minimizers.
\end{abstract}

\date{\today}

\maketitle

\setcounter{tocdepth}{1}
\tableofcontents

\section{Introduction}

It is a well-known experimental fact that highly negative ions do not exist: in fact, a neutral atom can bind at most one or two additional electrons. Heuristically, if a neutral atom has too many extra electrons, then the outermost electron has no electrostatic favor to stay together with the rest of the system{\blue,} which will have a negative net charge. However, deriving this fact rigorously from the first principles of quantum mechanics is a longstanding open problem, often referred to as the {\em ionization conjecture}, see for example \cite[Problem 9]{Simon-00} and \cite[Chapter 12]{LieSei-10}. We refer to \cite{Sigal-82,Ruskai-82,Lieb-84,LieSigSimThi-88,FefSec-90,SecSigSol-90,Nam-12} for partial results in the full many-body quantum theory.  The ionization problem has {\blue also been} studied (and solved) in several approximate models such as the Thomas-Fermi theory \cite{LieSim-77b}, the Thomas-Fermi-Dirac theory \cite{Benguria-79, Lieb-81b}, the Thomas-Fermi-von Weizs\"acker theory \cite{BenBreLie-81,BenLie-84} and the Hartree-Fock theory \cite{Solovej-91,Solovej-03}.

In this paper, we study the ionization problem in the Thomas-Fermi-Dirac-von Weizs\"acker (TFDW) theory. We consider the variational problem
\bq \label{eq:variational-problem}
I_V(m)=\inf \left\{ \cE_V(u)\,:\, u\in H^1(\R^3), \int_{\R^3} |u(x)|^2 \d x = m  \right\}
\eq
where $\cE_V$ is the TFDW energy functional
\begin{multline*}
\cE_V(u) = \int_{\R^3} \left( c_{\rm TF} |u(x)|^{10/3} - c_{\rm D} |u(x)|^{8/3} + c_{\rm W} |\nabla u(x)|^2 +V(x)|u(x)|^2 \right) \d x \\ + \frac{1}{2} \iint_{\R^3\times \R^3} \frac{|u(x)|^2 |u(y)|^2}{|x-y|} \d x \d y 
\end{multline*}
and
\bq \label{eq:def-V}
V(x)= -\sum_{j=1}^J \frac{Z_j}{|x-{\bf r}_j|}.
\eq
The functional $\cE_V(u)$  models the energy of a system of $m$ quantum electrons interacting with $J$ classical nuclei fixed at positions $\{{\bf r}_j\}_{j=1}^J \subset \R^3$. Here{\blue ,} $|u(x)|^2$ is interpreted as the electron density. The nuclear charges $\{Z_j\}$ and the total number of electrons $m$ are nonnegative numbers, which are not necessarily integers. The constants $c_{\rm TF}, c_{\rm D}$ and $c_{\rm W}$ are positive numbers, whose precise values are not important in our analysis. 

The first and the last terms in $\cE_V(u)$ are the semiclassical approximations of the kinetic energy and the self-interaction energy of electrons. They were introduced independently by Thomas \cite{Thomas-27} and Fermi \cite{Fermi-27} in 1927 in their celebrated theory for the electron distribution in atoms and molecules. Although the Thomas-Fermi theory captures the precise leading order of the ground state energy of large systems \cite{LieSim-77b}, it has some serious defects, most notably Teller's no-binding theorem \cite{Teller-62} and the absence of negative ions. Therefore, further corrections are necessary. In 1930, Dirac \cite{Dirac-30} proposed the correction $-c_{\rm D}\int |u|^{8/3}$ which models the exchange energy and in 1935, von Weizs\"acker suggested the correction $c_{\rm W} \int |\nabla u|^2$ to the kinetic energy. When only one of these two corrections is taken into account, we are left with the Thomas-Fermi-Dirac and the Thomas-Fermi-von Weizs\"acker theories. The simultaneous appearance of two corrections makes the TFDW theory more precise but also more difficult to analyze than its ancestors. We refer to \cite{Lieb-81b} for a pedagogical introduction to these density functional theories and their connections to many-body quantum mechanics.

One of the most fundamental {\blue questions} in the TFDW theory is to estimate the values of the parameters for which the variational problem  \eqref{eq:variational-problem} has minimizers. In 1987, Lions \cite{Lions-87} proved the existence of minimizers when 
$$m \le \sum_{j=1}^J Z_j$$
(see also \cite{Bris-93}). However, the ionization problem, which corresponds to the \emph{nonexistence} of minimizers when $m$ is large,  remains mostly open. In fact, the nonexistence when $V\equiv 0$ is already surprisingly delicate and has been solved recently by Lu and Otto \cite{LuOtto-14}.  Their proof is based crucially on the translation-invariance of $\cE_0(u)$ and does not apply to the general case. 

In this paper, we will establish the nonexistence of minimizers for  \eqref{eq:variational-problem} with an external potential $V$. Our main result is the following.

\begin{theorem}[Nonexistence with small nuclear charges] \label{thm:small-Z} There exist constants $Z_{\rm c}>0$ and $M_{\rm c}>0$ such that the variational problem $I_V(m)$ in \eqref{eq:variational-problem} with $V$ in \eqref{eq:def-V} has no minimizers when
$$
\sum_{j=1}^J Z_j \le Z_{\rm c} \quad \text{and} \quad m \ge M_{\rm c}. 
$$
\end{theorem}

Our overall strategy is to show that if $I_V(m)$ has a minimizer with $m$ sufficiently large, then all but $a \le \sum_{j=1}^J Z_j$ electrons must escape to infinity. The energy of {\blue the} electrons that {\blue stay} is not smaller than $I_V(a)$, while the energy of {\blue the} electrons that {\blue escape} is not smaller than $I_0(m-a)$. Since this splitting of the energy will only hold up to an error term, we need a quantitative version of the strict binding inequality $I_V(m)<I_V(a)+I_0(m-a)$ in order to obtain a contradiction. In this last step we have to assume the smallness of the nuclear charges. We hope to be able to remove this technical assumption in the future.  

As by-products of our proof, we obtain some related results. First, the nonexistence can be extended easily to any external potential which decays faster than the Coulomb potential. 

\begin{theorem}[Nonexistence with short-range potentials] \label{thm:short-range} Assume that
$$ V\in L_{\rm loc}^{3/2}(\R^3),\quad V(x) \le 0 \quad \text{and}\quad \lim_{|x|\to \infty} |x| V(x) =0.$$
Then there exists $M_{\rm c}'> 0$ such that the variational problem $I_V(m)$ in \eqref{eq:variational-problem} has no minimizer when $m \geq M_{\rm c}'$.
\end{theorem}

Next, for general Coulomb-type external potentials, including the molecular form \eqref{eq:def-V} with arbitrary nuclear charges, we obtain two weak forms of the nonexistence. 

\begin{theorem} [Nonexistence of stable and radial minimizers] 
\label{thm:no-stable-minimizers} Assume that 
$$ V\in L_{\rm loc}^{3/2}(\R^3), \quad V(x)\in \mathbb{R} \quad \text{and}\quad \limsup_{|x|\to \infty} |x V(x)| <\infty.$$
\noindent {\rm (i)} There exist $a>0$ and $b>0$ such that both $I_V(a)$ and $I_0(b)$ have minimizers, and 
\[
 I_V(a+b) = I_V(a) + I_0(b).
\]

\noindent {\rm (ii)} There exists a constant $M_{\rm r}>0$ such that for all $m \ge M_{\rm r}$, a minimizer of the variational problem $I_V(m)$ in \eqref{eq:variational-problem}, if it exists, is not radially symmetric. 
\end{theorem}

Theorem \ref{thm:no-stable-minimizers} (i) corresponds to the nonexistence of stable minimizers. From the physical point of view, the equality $I_V(a+b) = I_V(a) + I_0(b)$ implies that the minimizers of $I_V(a+b)$, if they exist, are not stable because $b$ particles may escape to infinity without increasing the energy.

Theorem \ref{thm:no-stable-minimizers} (ii) is related to the fact that in the atomic case, $V(x)=-Z/|x|$,  all minimizers are  radially symmetric decreasing if $m\le Z$, see  \cite[Theorem 8.6]{Lieb-81b}. Indeed, in the atomic case, it was conjectured in \cite{Lieb-81b} that any minimizer, if it exists, must be radially symmetric. This remains an open problem. 

\subsection*{Note added in proof} In the atomic case, $V(x)=-Z/|x|$, it has been proved very recently in \cite{FraNamBos-16} that \eqref{eq:variational-problem} has no minimizer if $m>Z+C$, where $Z>0$ is arbitrary and $C>0$ is independent of $Z$.  This result is much stronger than our result in the present paper. However, the method we presented below is different from that of \cite{FraNamBos-16} and it works for a more general class of external potentials $V$ (including the molecular case). Therefore, we hope that the approach in the present paper still has some independent interest.

%
%
%
%
%

\subsection*{Outline of the paper} In Section~\ref{sec:general} we establish some basic properties of the energy functional and its minimizers. Then we quickly revisit the existence and nonexistence results for $I_0(m)$ in Section~\ref{sec:I0} -- the proof of the existence part is deferred to the appendix. The main new part of the paper starts with Section~\ref{sec:radius} where we provide a detailed study of the \emph{radius} of the minimizers. The obtained estimates will be used in Section \ref{sec:short-range} to prove Theorem \ref{thm:short-range} for short-range potentials. In Section \ref{sec:improved}, we establish an improved radius estimate, which is essential to deal with Coulomb-type potentials. The proofs of Theorems \ref{thm:small-Z} and \ref{thm:no-stable-minimizers} are presented in Sections \ref{sec:proof-small-Z} and \ref{sec:instability}, respectively.


\subsection*{Acknowledgements} 
We would like to thank Rafael Benguria for motivating discussions. We thank Jianfeng Lu and Felix Otto for helpful correspondence on the translation-invariant case
and for pointing out the reference \cite{SanSol-04b} to us. P.T. Nam is supported by the Austrian Science Fund (FWF) under Project Nr. P 27533-N27. H. Van Den Bosch acknowledges support from CONICYT (Chile) through CONICYT--PCHA/Doctorado Nacional/2014 and through Fondecyt Project \# 112-0836, and from the Iniciativa Cient\'ifica Milenio (Chile) through the Millenium Nucleus RC--120002 ``F\'isica Matem\'atica''.

\section{General estimates} \label{sec:general}
The goal of this section is to establish some properties of the energy functional and its minimizers that hold independently of the detailed properties of the external potential $V$. We will assume throughout this section that
$$
 V\in L_{\rm loc}^{3/2}(\R^3), \quad  V \le 0 \quad \text{ and} \quad \lim_{|x|\to \infty} V(x)=0 .
$$
The Coulomb interaction energy will be written as
\[
D(f,g) = \frac{1}{2} \iint_{\R^3 \times \R^3} \frac{   \overline{f(x)} g(y)}{\abs{x-y}} \d x \d y.
\]
We start with some basic properties of the energy functional. 
 
\begin{lemma}[Basic energy estimate] \label{lem:basic-energy} For all $u\in H^1(\R^3)$ we have 
\bq \label{eq:kinetic-estimate}
\cE_V(u) + C_1\int |u|^2  \ge \int \Big(  \frac{c_{\rm TF}}{2}  |u|^{10/3}+ \frac{c_{\rm W}}{2}  |\nabla u|^2 +  |V||u|^2 \Big) 
+ D(\abs{u}^2, \abs{u}^2).
\eq
with $$C_1=\frac{c_{\rm D}^2}{2 c_{\rm TF}} - \frac{1}{2} \inf {\rm spec}\left(-c_{\rm W} \Delta - 4|V| \right) .$$
In the molecular case, where $V$ is given by \eqref{eq:def-V}, we have 
$$C_1 \le \frac{c_{\rm D}^2}{2 c_{\rm TF}} +  \frac{2}{c_{\rm W}}\Bigl(\sum_{j=1}^J Z_j \Bigr)^2.$$ 
\end{lemma}
\begin{proof} Since $V\in L^{3/2}(\R^3)+ L^\infty(\R^3)$, by Sobolev's inequality we find that the Schr\"odinger operator $- c_{\rm W}\Delta - 4|V|$ is bounded from below on $L^2(\R^3)$, see \cite[Section 11.3, Eq. (15)]{LieLos-01}. Using
$$\frac{c_{\rm TF}}{2}|u|^{10/3}-c_{\rm D} |u|^{8/3} = \left(\sqrt{\frac{c_{\rm TF}}{2}} |u|^{5/3}- \frac{c_{\rm D} |u|}{\sqrt{2 c_{\rm TF}}}\right)^2 - \frac{c_{\rm D}^2}{2 c_{\rm TF}} |u|^2 \ge - \frac{c_{\rm D}^2}{2 c_{\rm TF}} |u|^2$$
we obtain \eqref{eq:kinetic-estimate} with 
$$C_1=\frac{c_{\rm D}^2}{2 c_{\rm TF}} - \frac{1}{2} \inf {\rm spec}\left(-c_{\rm W} \Delta - 4|V| \right) .$$
In the molecular case \eqref{eq:def-V}, by using the hydrogen bound $-\Delta-s/|x| \ge -s^2/4$, we obtain 
$$
\inf {\rm spec}\left(-c_{\rm W} \Delta - 4|V| \right) \ge - \frac{4}{c_{\rm W}}\Bigl(\sum_{j=1}^J Z_j \Bigr)^2,
$$
and hence
$$C_1 \le \frac{c_{\rm D}^2}{2 c_{\rm TF}} + \frac{2}{c_{\rm W}}\Bigl(\sum_{j=1}^J Z_j \Bigr)^2.$$
\end{proof}

 Lemma \ref{lem:basic-energy} implies that $I_V(m)\ge -C_1 m$ and it also leads to some a-priori estimates on minimizers.  Refined estimates on minimizers will be obtained from the following binding inequality, which is a typical ingredient of the concentration-compactness method \cite{Lions-84,Lions-84b}.

\begin{lemma}[Binding inequality] \label{lem:binding} For all $m>0$ we have 
\bq
\label{eq:non-strict-bind-ineq} I_V(m) \le I_V(m')+ I_0(m-m'),\quad \forall 0 \le m' \le m. 
\eq
Moreover, $I_V(m)\le I_0(m)<0$ for all $m>0$ and the function $m\mapsto I_V(m)$ is strictly decreasing and continuous.
\end{lemma}

\begin{proof} Let us take two smooth functions $v_1,v_2$ with compact supports such that $\int |v_1|^2=m'$ and $\int |v_2|^2=m-m'$. For any vector $x_0\in \R^3\minus \{0\}$, one has  
$$ I_V(m) \le \lim_{R\to \infty} \cE_V\bigl(v_1 (.) + v_2 (. + R x_0) \bigr) = \cE_V(v_1) + \cE_0(v_2).$$
Here we have used the fact that $V(x)$ and the Coulomb potential $|x|^{-1}$ vanish at infinity. Optimizing the right-hand-side over $v_1$ and $v_2$ gives \eqref{eq:non-strict-bind-ineq}. 

From \eqref{eq:non-strict-bind-ineq}, we have immediately that $I_V(m)\le I_0(m)$. The strict inequality $I_0(m)<0$ can be seen by choosing a test function $v$ such that 
$$
\int |v|^2=m, \quad  D(\abs{v}^2, \abs{v}^2) < c_{\rm D} \int |v|^{8/3}.
$$ 
Then, consider the trial function $v_\ell=\ell^{3/2}v(\ell \dotv)$.
Since the kinetic terms scale as $\ell^2$ and the electrostatic terms scale as $\ell$, the energy becomes strictly negative when $\ell>0$ is small enough. 

Combining \eqref{eq:non-strict-bind-ineq} with $I_0(m) < 0$ when $m>0$, we obtain that  $m \mapsto I_V(m)$ is strictly decreasing. The continuity of $m\mapsto I_V(m)$ follows from a standard argument based on the variational principle and appropriate trial states. 
\end{proof}

Using the previous lemma, we obtain a key estimate on minimizers. 

\begin{lemma}[Basic localization estimate] \label{lem:basic-loc} Assume that $I_V(m)$ has a minimizer $u_m$. Consider a partition of unity consisting of smooth functions $\chi,\eta:\R^3\to [0,1]$ satisfying $\chi^2+\eta^2 \equiv 1$ on $\R^3$, and define $\Omega \equiv \{x \in \R^3 | \chi(x) \in (0,1)\}$. Then,
\bq \label{eq:loc}
 2 D(|\chi u_m|^2, |\eta u_m|^2)   \le - \int_{\R^3} V|\eta u_m|^2  + C_2\int_{\Omega} |u_m|^2,
\eq
where
$$
C_2= \frac{c_{\rm D}^2}{4c_{\rm TF}} + c_{\rm W} \Big( \|\nabla \chi\|^2_{L^\infty}+ \|\nabla \eta\|^2_{L^\infty} \Big) .
$$
\end{lemma}

\begin{proof} Since $u_m$ is a minimizer for $I_V(m)$, by applying \eqref{eq:non-strict-bind-ineq}, we find that
$$
\cE_V(u_m)  - \cE_V(\chi u_m) - \cE_0(\eta u_m) \le 0.
$$
It remains to show that 
\begin{align} \label{eq:loc-0}
\cE_V(u_m)  - \cE_V(\chi u_m) - &\cE_0(\eta u_m) \ge  2 D(|\chi u_m|^2, |\eta u_m|^2)    \\
&  +  \int_{\R^3} V(x)|\eta (x) u_m(x)|^2 \d x - C_2 \int_{\Omega} |u_m(x)|^2 \d x .\nn
\end{align}
Indeed, by using 
$$1- \chi^2(x) \chi^2(y) - \eta^2(x) \eta^2(y)=\chi^2(x) \eta^2(y)+ \eta^2(x)\chi^2(y)$$
and interchanging the variables $x$ and $y$, we get
\begin{align}
&\frac{1}{2} \iint_{\R^3\times \R^3}  \frac{|u_m(x)|^2|u_m(y)|^2}{|x-y|} \Big( 1- \chi^2(x) \chi^2(y) - \eta^2(x) \eta^2(y) \Big) \d x \d y \nn\\
&= \iint_{\R^3\times \R^3}  \frac{|\chi(x) u_m(x)|^2|\eta(y) u_m(y)|^2}{|x-y|} \d x \d y . \label{eq:loc-1}
\end{align}
Next, from the IMS localization formula 
$$
\abs{\nabla u}^2 - \abs{\nabla (\chi u)}^2- \abs{\nabla (\eta u)}^2 = -\bigl(\abs{\nabla \chi}^2+\abs{\nabla \eta}^2\bigr) \abs{u}^2 ,
$$
it follows that
\begin{align} \label{eq:loc-2}
&c_{\rm W}\int_{\R^3} \Big(  |\nabla u_m|^2 - | \nabla (\chi u_m) |^2 - | \nabla (\eta u_m) |^2  \Big) \nn\\
&=  - c_{\rm W} \int_{\R^3} \Big(  |\nabla \chi|^2 + |\nabla \eta|^2\Big) |u_m|^2 \nn\\
&\ge - c_{\rm W} \Big(\|\nabla \chi\|^2_{L^\infty}+ \|\nabla \eta\|^2_{L^\infty} \Big) \int_{\Omega} |u_m|^2.
\end{align}
Finally, using
\[ 
0 \le 1- \chi^{8/3}-\eta^{8/3} \le 1- \chi^{10/3}-\eta^{10/3} \le \1_{\Omega} 
\]
and 
\begin{equation}
 \label{eq:complete_square}
c_{\rm TF} |u_m|^{10/3} - c_{\rm D} |u_m|^{8/3}  \ge - \frac{c_{\rm D}^2}{4c_{\rm TF}}|u_m|^2,
\end{equation}
we can estimate
\begin{align} \label{eq:loc-3}
 &\int_{\R^3} \Big( c_{\rm TF}  (1-\chi^{10/3}-\eta^{10/3}) |u_m|^{10/3} - c_{\rm D}  (1-\chi^{8/3}-\eta^{8/3}) |u_m|^{8/3} \Big)  \nn \\
 &\ge \int_{\R^3} (1- \chi^{10/3}-\eta^{10/3}) \Big( c_{\rm TF} |u_m|^{10/3} - c_{\rm D} |u_m|^{8/3} \Big) \ge - \frac{c_{\rm D}^2}{4c_{\rm TF}} \int_{\Omega} |u_m|^2. 
\end{align}
Putting \eqref{eq:loc-1}, \eqref{eq:loc-2} and \eqref{eq:loc-3} together, we obtain \eqref{eq:loc-0} and complete the proof. 
\end{proof}

We will frequently apply the localization estimate in Lemma \ref{lem:basic-loc} to separate the energy contribution inside and outside a ball centered at the origin. The corresponding localization functions are defined below.

\begin{definition}[Standard localization functions]
 Fix two smooth functions $f, g:\R \to [0,1]$ such that 
\bq \label{eq:def-fg}
f^2+g^2 \equiv 1 \text{ on } \R, \quad f(t)=1 \text{ when } t\le 0, \quad f(t)=0 \text{ when }  t\ge 1. 
\eq
We can choose $f$, $g$ such that $\abs{f'} \leq 2$, $\abs{g'} \leq 2$. For every $R>0$, we define 
\bq \label{eq:def-chiR-etaR}
\chi_R (x)= f(|x|-R),\quad \eta_R (x)=g (|x|-R).
\eq
Then $\chi_R^2+\eta_R^2\equiv 1$ on $\R^3$, $\chi_R(x) =1$ when $|x|\le R$ and $\chi_R(x)=0$ when $|x| \ge R+1$. Moreover, $\|\nabla \chi_R\|_{L^\infty} \le 2$ and $\|\nabla \eta_R\|_{L^\infty}\le 2$. 
\end{definition}

Using these localization functions in Lemma~\ref{lem:basic-loc}, we obtain

\begin{lemma}[Annulus estimate] \label{lem:annulus-est} Assume that $I_V(m)$ has a minimizer $u_m$. Then for all $R \ge 1$,
\begin{align} \label{eq:annulus-estimate-1}
 \int\limits_{|x|\le R}|u_m(x)|^2 \d x  \int\limits_{|y|\ge 2R}|u_m(y)|^2 \d y  \le 12   \int\limits_{|x|\ge R} \left( C_3+ |xV(x)| \right) |u_m(x)|^2 \d x,
\end{align}
where $C_3 = 8 c_{\rm w} + c_{\rm D}^2/(4c_{\rm TF})$.
\end{lemma}

\begin{proof} 
For every $R\ge 1$, we apply \eqref{eq:loc} with $\chi=\chi_R$ and $\eta=\eta_R$. Using the triangle inequality, $|x-y|\le |x|+|y| \le 3|y|$ when $|y| \ge \max\{1, |x|-1\}$, so we have
\begin{align} \label{eq:annulus-0}
& \frac{1}{3} \left( \int_{\R^3} |\chi_R(x)u_m(x)|^2 \d x \right) \left( \int_{\R^3} \frac{|\eta_R (y)u_m(y)|^2}{|y|}   \d y \right)  \nn \\
&\le  - \int_{\R^3} V(x) |\eta_R(x) u_m(x)|^2 \d x + C_3 \int_{R\le |x| \le R+1} |u_m(x)|^2  \d x.
\end{align}
By replacing $R$ with $R+k$ in \eqref{eq:annulus-0} and taking the sum over $k=0,1,2,...$, we obtain
\begin{multline} \label{eq:sumRk}
\sum_{k=0}^\infty \frac{1}{3} \left( \int_{\R^3} |\chi_{R+k}(x)u_m(x)|^2 \d x \right) \left( \int_{\R^3} \frac{|\eta_{R+k} (y)u_m(y)|^2}{|y|}   \d y \right)   \\
\le - \sum_{k=0}^\infty \int_{\R^3} V(x) |\eta_{R+k}(x) u_m(x)|^2 \d x + C_3 \int_{|x| \ge R} |u_m(x)|^2 \d x .
\end{multline}
Let us estimate the left hand side of \eqref{eq:sumRk}. For the first factor, we use the uniform bound  
$|\chi_{R+k}(x)|^2  \ge \1({\scriptstyle|x|\le R})$. For the second factor, note that
\begin{align*}
 \sum_{k=0}^\infty \eta^2_{R+k}(y) \geq \sum_{k=0}^\infty \1({\scriptstyle|y|\ge R+k+1}) \geq \frac{1}{2}(\abs{y}-R) \1({\scriptstyle\abs{y} \geq R+1}),
\end{align*}
which gives
\begin{align*}
 \frac{1}{\abs{y}}\sum_{k=0}^\infty \eta^2_{R+k}(y) \geq \frac{1}{2}\left(1- \frac{R}{\abs{y}} \right)\1({\scriptstyle\abs{y} \geq R+1})  
 \geq \frac{1}{4} \1({\scriptstyle\abs{y} \geq 2R}).
 \end{align*}
Combining these inequalities,
\begin{multline} \label{eq:sumRk-left}
\sum_{k=0}^\infty \frac{1}{3} \left( \int_{\R^3} |\chi_{R+k}(x)u_m(x)|^2 \d x \right) \left( \int_{\R^3} \frac{|\eta_{R+k} (y)u_m(y)|^2}{|y|}   \d y \right)\\
 \ge \frac{1}{12} \left( \int_{|x|\le R}|u_m(x)|^2 \d x \right) \left( \int_{|y|\ge 2R} |u_m(y)|^2 \d y\right).
\end{multline}
To bound the right hand side of \eqref{eq:sumRk}, we use
\begin{align*}
\sum_{k=0}^\infty\eta^2_{R+k}(x) &\le \sum_{k=0}^\infty \1({\scriptstyle |x|\ge R+k}) \le (\abs{x}-R+1)\1({\scriptstyle|x|\ge R}) \le \abs{x} \1({\scriptstyle|x|\ge R})
\end{align*}
and deduce that 
\begin{align}
- \sum_{k=0}^\infty \int_{\R^3} V(x)|\eta_{R+k}u_m(x)|^2 \d x \le  \int_{|x|\ge R} |xV(x)|\,|u_m(x)|^2 \d x\label{eq:sumRk-right}.
\end{align}
Substituting \eqref{eq:sumRk-left} and \eqref{eq:sumRk-right} into \eqref{eq:sumRk}, we obtain \eqref{eq:annulus-estimate-1}. 
\end{proof}

\section{Existence and nonexistence for $I_0(m)$} \label{sec:I0}

In this section, we revisit some well-known properties of the translation-invariant problem $I_0(m)$, which will be used in the rest of the paper. First, let us quickly prove the nonexistence of minimizers when $m$ is large, recovering the main result in \cite{LuOtto-14}.  

\begin{lemma}[Nonexistence for $I_0(m)$]\label{lem:nonexistence-I0} $I_0(m)$ has no minimizers when $m$ is sufficiently large. 
\end{lemma}
\begin{proof} We will denote by $C$ a generic constant independent of $m$, the value of which may change from line to line. Assume $I_0(m)$ has a minimizer $u_m$. For every $R>0$, we define
$$ M_R  = \sup_{y\in \R^3} \int_{|x-y|\le R} |u_m(x)|^2 \d x.$$
Since $R\mapsto M_R$ is increasing and $0=M_0\le M_{\infty}=m$, there exists $r_m$ such that
$
M_{r_m} >m^{2/3} \ge M_{r_m/2} . 
$
By the definition of $M_{r_m}$, there exists $y_m\in \R^3$ such that 
\bq \label{eq:I0-rm-0}
\int_{|x-y_m|\le r_m} |u_m(x)|^2 \d x \ge m^{2/3} .
\eq
Since $\cE_0(u)$ is translation-invariant, by replacing $u_m$ with $u_m(\dotv +y_m)$ we can assume $y_m=0$. From Lemma~\ref{lem:annulus-est} and the fact that $\cE_0(u_m)=I_0(m)<0$ we find that 
$$
\int |u_m(x)|^{10/3}  \d x \le Cm.
$$
Therefore, by H\"older's inequality, we get  
\begin{align} \label{eq:Holder-rm}m^{2/3} \le \int_{|x|\le r_m} |u_m(x)|^2 \d x  & \le \Big(\int_{|x|\le r_m} 1 \d x \Big)^{2/5} \Big(\int_{|x|\le r_m} |u_m(x)|^{10/3}\d x \Big)^{3/5} 
\nn\\
&\le C r_m^{6/5} m^{3/5}.
\end{align} 
Thus $r_m \ge 1$ when $m$ is sufficiently large. By applying \eqref{eq:annulus-estimate-1} for $R=r_m$ and using \eqref{eq:I0-rm-0} (with $y_m=0$), we find that
\bq \label{eq:I0-contra-1}
\int_{|x|\ge 2r_m} |u_m(y)|^2 \d y \le 12 C_3 m^{1/3}.
\eq
On the other hand, it is easy to see that there is a universal constant $C_B$ such that the ball $B(0,2r_m)\subset \R^3$ can be covered by $C_B$ smaller balls of radius $r_m/2$. By the definition of $M_{r_m/2}$, the integral of $|u_m|^2$ over each smaller ball is smaller than $M_{r_m/2}$, and $ M_{r_m/2} \le m^{2/3}$.  Therefore, 
\bq \label{eq:I0-contra-2}
\int_{|x|\le 2r_m} |u(x)|^2 \d x \le C_B m^{2/3}.
\eq
Combining \eqref{eq:I0-contra-1} and \eqref{eq:I0-contra-2}, we find that 
$$ m= \int_{\R^3} |u(x)|^2 \d x \le (C_B + 12 C_3) (m^{1/3}+ m^{2/3}). $$
Thus $m$ is bounded by a universal constant. 
\end{proof}

We will also need the following existence result, which is a typical application of the concentration-compactness method \cite{Lions-84,Lions-84b}. Since we could not localize a precise reference, a proof will be provided in the appendix. 

\begin{lemma}[Existence for $I_0(m)$] \label{lem:mass-decomp-I0}
For all $m>0$, the followings hold true. 

\medskip

\noindent {\rm (i)} If $\{v_n\}$ is a minimizing sequence for $I_0(m)$, then up to subsequences and translations, $v_n$ converges weakly in $H^1(\R^n)$ to some $v \not\equiv 0$.

\medskip

\noindent {\rm (ii)} We can decompose 
$$m=\sum_{j=1}^\infty m_j \quad{\text and} \quad I_0(m)=\sum_{j=1}^\infty I_0(m_j),$$
where $m_j \ge 0$ and $I_0(m_j)$ has a minimizer for all $j\ge 1$.

\medskip

\noindent {\rm (iii)} There exists $m_0 > 0$ such that 
$$
I_0(m)< I_0(m')+ I_0(m-m') \quad \text{for all } \quad 0<m'< m\le m_0.
$$
Consequently, $I_0(m)$ has a minimizer for all $m\le m_0.$

\end{lemma}

\section{Radius estimates} \label{sec:radius}
Throughout this section we will assume that $I_V(m)$ has a minimizer $u_m$ for $m$ large. Our goal is to obtain several estimates on how the mass of $u_m$ is distributed that will be used afterwards to derive a contradiction. We will assume that
$$
 V\in L_{\rm loc}^{3/2}(\R^3), \quad V(x)\in \mathbb{R} \quad \text{ and} \quad \limsup_{|x|\to \infty} |xV(x)|<\infty.
$$

\subsection*{Notations} We will always denote by $C$ a generic (typically large) constant independent of $V$ and $m$. In addition, we will denote by $C_V$ a generic constant dependent on $V$ but independent of $m$. In the molecular case \eqref{eq:def-V}, $C_V$ can be chosen to be $C(1+\sum_{j=1}^J Z_j)$. 
In this way we keep the notations concise while keeping track of the dependence of constants on the nuclear charges, as this will be required for the proof of Theorem~\ref{thm:small-Z}.

\begin{definition}[Radius of the system]
Define $\chi_R$ and $\eta_R$ as in \eqref{eq:def-chiR-etaR}.
We define $R_m > 0$ such that
$$
\int_{\R^3} |\chi_{R_m} u_m|^2 = \int_{\R^3} |\eta_{R_m} u_m|^2 = \frac{m}{2}.
$$
Since $R\mapsto \int |\chi_R u_m|^2$ is continuous and increases from $0$ to $m$, this is always possible.
\end{definition}

\begin{lemma} \label{lem:Rm>>1} We have $R_m \ge m^{1/3}/C_V -1.$
\end{lemma}
\begin{proof} The proof is similar to the one of \eqref{eq:Holder-rm}. Since $\cE_V(u_m) = I_V(m)<0$, It follows from \eqref{eq:kinetic-estimate} that 
\bq \label{eq:u10/3}
\int |u|^{10/3} \le  C_V^2 m.
\eq
Therefore, by H\"older's inequality,
\begin{align*}
\frac{m}{2} = \int \chi_{R_m} |u_m|^2 &\le \left( \int |\chi_{R_m}|^{5/2} \right)^{2/5} \left( \int |u_m|^{10/3} \right)^{3/5} \\
&\le \Big( \frac{4\pi}{3}  (R_m+1)^3 \Big)^{2/5}(C_V^2 m)^{3/5}
\end{align*}
and the desired estimate follows. 
\end{proof}

\begin{lemma} \label{lem:>2Rm} When $m \ge C_V^3$ we have
 \bq \label{eq:>2Rm}
\int_{R_m/2 \le |x| \le 2R_m+2}|u_m(x)|^2\d x \ge m-C_V.
\eq
\end{lemma}

\begin{proof} If $m$ is larger than $C_V^3$, then $R_m$ is large by Lemma~\ref{lem:Rm>>1},  and hence $|xV(x)|\le C_V$ for $|x|\ge R_m/2$ by the assumption on $V$. Applying Lemma~\ref{lem:annulus-est}, we find that
\begin{align} \label{eq:annulus-estimate-1b} 
\left( \int_{|x|\le R}|u_m(x)|^2 \d x \right) \left( \int_{|y|\ge 2R}|u_m(y)|^2 \d y\right) \le C_V m 
\end{align}
for all $R \ge R_m/2$. We use \eqref{eq:annulus-estimate-1b} with $R=R_m/2$ and observe that
$$
 \int_{|y|\ge R_m}|u_m(y)|^2\d y \ge \int_{\R^3}|\eta_{R_m}(y)u_m(y)|^2\d y = \frac{m}{2} .
$$
This gives 
\bq \label{eq:>2Rm-b}
 \int_{|x|\le R_m/2}|u_m(y)|^2\d x \le C_V.  
\eq
Similarly, using \eqref{eq:annulus-estimate-1b} with $R=R_m+1$ and noting  
$$
 \int_{|x|\le R_m+1}|u_m(y)|^2\d x \ge \int_{\R^3}|\chi_{R_m}(x)u_m(x)|^2\d x = \frac{m}{2}, 
$$
we obtain
\bq \label{eq:>2Rm-a}
 \int_{|x|\ge 2R_m+2}|u_m(x)|^2\d x \le C_V.  
\eq
Putting  \eqref{eq:>2Rm-a} and \eqref{eq:>2Rm-b} together, we obtain \eqref{eq:>2Rm}.
\end{proof}

\begin{lemma} \label{lem:Rm>m} When $m \ge C_V^3$, we have $R_m \ge m/C.$
\end{lemma}
\begin{proof}
 Let us apply \eqref{eq:annulus-0} with $R=R_m$:
\begin{multline} \label{eq:Rm>=m-1}
\frac{1}{3}\left( \int_{\R^3} |\chi_{R_m}(x) u_m(x)|^2 \d x\right)  \left( \int_{\R^3}\frac{|\eta_{R_m}(y) u_m(y)|^2}{|y|} \d y\right) 
\\ \le - \int_{\R^3} V(x) |\eta_{R_m}(x) u_m(x)|^2 \d x  + C \int_{R_m\le |x|\le R_m+1} |u_m(x)|^2 \d x.
\end{multline}
Using Lemma~\ref{lem:>2Rm}, we have 
\begin{align} \label{eq:Rm>=m-2}
\int_{\R^3}\frac{|\eta_{R_m}(y) u_m(y)|^2}{|y|} \d y &\ge \frac{1}{2R_m+2}  \int_{ |x| \le 2R_m+2} |\eta_{R_m}(y) u_m(y)|^2 \d y \nn \\
& = \frac{1}{2R_m+2} \left( \frac{m}{2} - \int_{|x|\ge 2R_m+2} |u_m(x)|^2 \d x \right)  \nn\\
& \ge \frac{m-C_V}{4R_m+4}.
\end{align}
Moreover, 
\begin{align}  \label{eq:Rm>=m-3}
 - \int_{\R^3} V(x) |\eta_{R_m}(x) u_m(x)|^2 \d x 
         & \le \sup_{|y|\ge R_m} |yV(y)| \int_{|x|\ge R_m} \frac{|u_m(x)|^2}{|x|} \d x  \nn\\
         & \le \frac{C_V m}{R_m}.
\end{align}
Here we have used again the assumption that $|xV(x)|\le C_V$ for $|x|$ large. Substituting \eqref{eq:Rm>=m-2} and \eqref{eq:Rm>=m-3} into \eqref{eq:Rm>=m-1} and using the obvious bound 
$$\int_{R_m\le |x|\le R_m+1} |u_m(x)|^2 \d x \le \int_{\R^3} |u_m|^2 = m$$
we deduce that 
\begin{align*}
  \frac{1}{3} \cdot \frac{m}{2} \cdot \frac{m-C_V}{4R_m+4} \le \frac{C_V m}{R_m} + C m \\
\end{align*}
which implies $R_m \ge m/C$.
\end{proof}

The previous estimate tells us that, if a minimizer exists for large $m$, most of the electrons are far away from the origin.
In the next lemma we use this fact to approximately separate the energy in a term coming from the electrons that remain close to the origin and a term coming from those that are far away. 

\begin{lemma} \label{lem:rm-am} When $m\ge C_V^3$, there exists $r_m \in [R_m m^{-1/2}, 2R_m m^{-1/2}]$ such that 
\bq \label{eq:am<=Z}
a_m:=\int_{\R^3} |\chi_{r_m}u_m|^2 \le  \left(  1 + Cm^{-1/2}\right) \Big( \sup_{|x|\ge r_m} |xV(x)|  + C_V m^{-1/2} \Big),
\eq
and
\bq \label{eq:local-rm-am}
\cE_V(u_m)- \cE_V(\chi_{r_m} u_m) - \cE_0(\eta_{r_m} u_m ) \ge -\sup_{|x|\ge r_m} |xV(x)| \frac{2m}{R_m} - \frac{C_V m^{1/2}}{R_m}.
\eq
\end{lemma}
\begin{proof} Recall that by Lemma \ref{lem:Rm>m}, $R_m m^{-1/2} \ge m^{1/2}/C$. Since
 \[
 \int_{|x| \le R_m/2} \abs{u_m(x)}^2 \d x \le C_V
 \]
by Lemma~\ref{lem:>2Rm},  there exists 
 $r_m \in [R_m m^{-1/2}, 2R_m m^{-1/2}]$ such that
 \[
 \int_{r_m \le \abs{x} \le r_m+1} \abs{u_m(x)}^2 \d x \le  \frac{C_V m^{1/2}}{R_m}.
 \]
Applying the localization estimate \eqref{eq:loc-0} to $R=r_m$, we have
\begin{align} \label{eq:loc-rm}
0 &\ge \cE_V(u_m)- \cE_V(\chi_{r_m} u_m) - \cE_0(\eta_{r_m} u_m ) \nn \\ 
	&\ge 2	D(|\chi_{r_m} u_m|^2, |\eta_{r_m}u_m|^2)  
	 +  \int_{\R^3} V(y)|\eta_{r_m} (y) u_m(y)|^2 \d y  - \frac{C_V m^{1/2}}{R_m}.
\end{align}
The first term can be bounded by using
 \[
 \abs{x-y} \le \abs{x}+\abs{y} = \abs{y} \left( \frac{|x|}{\abs{y}} + 1 \right) \le |y| (Cm^{-1/2}+1)
 \]
when $|x| \le r_m+1 \le R_m/2 \le |y|$, so
 \begin{align}
  & 2 D(|\chi_{r_m} u_m|^2, |\eta_{r_m}u_m|^2) \nn\\
   & \qquad \ge \int_{\R^3} |\chi_{r_m}(x) u_m(x)|^2 \left( \int_{\abs{y}\ge R_m / 2}   \frac{|\eta_{r_m}(y)u_m(y)|^2}{|x-y|}  \d y \right) \d x \nn\\
    & \qquad \ge a_m \left(  1 + Cm^{-1/2}\right)^{-1} \int_{\abs{y}\ge R_m / 2}\frac{|u_m(y)|^2}{|y|}  \d y. \label{eq:am-Duu}
 \end{align} 
 Moreover, 
 \begin{align}
   \int_{\R^3} |V(y)| |\eta_{r_m} (y) & u_m(y)|^2 \d y 
               \le \sup_{|x|\ge r_m} |xV(x)| \int_{\abs{y} \ge r_m} \frac{|u_m(y)|^2}{\abs y} \d y \nn \\
              &\le \sup_{|x|\ge r_m} |xV(x)|   \left( \frac{C_V}{r_m} +  \int_{ \abs{y} \ge R_m/2} \frac{|u_m(y)|^2}{\abs y} \d y \right) \label{eq:am-V},
 \end{align}  
where the last estimate follows from Lemma~\ref{lem:>2Rm} by noting
$$
\int_{r_m\le |y| \le R_m/2} \frac{|u_m(y)|^2}{|y|} \d y \le \frac{1}{r_m} \int_{|y|\le R_m/2} |u_m(y)|^2 \d y \le \frac{C_V}{r_m}.
$$
Inserting \eqref{eq:am-Duu} and \eqref{eq:am-V} into \eqref{eq:loc-rm}, we obtain
 \begin{align} \label{eq:a<Z-1}
 0 &\ge \cE_V(u_m)- \cE_V(\chi_{r_m} u_m) - \cE_0(\eta_{r_m} u_m ) \\
 &\ge  \left( \frac{a_m}{ 1 + Cm^{-1/2}} - \sup_{|x|\ge r_m} |xV(x)| \right) \int\limits_{|y|\ge R_m/2} \frac{|u_m(y)|^2}{|y|} \d y  - \frac{C_V m^{1/2}}{R_m}.  \nn  \end{align} 
On the other hand, using Lemma \ref{lem:>2Rm} we have
$$
\int_{|y|\ge R_m/2} \frac{|u_m(y)|^2}{|y|} \d y \ge \frac{1}{2R_m} \int_{R_m/2 \le |y| \le 2R_m}  |u_m(y)|^2 \d y \ge \frac{m-C_V}{2R_m}.  
$$
Thus we deduce from \eqref{eq:a<Z-1} that
\begin{align*}
 \frac{a_m}{ 1 + Cm^{-1/2}} &\le \sup_{|x|\ge r_m} |xV(x)| +  \frac{C_V m^{1/2}}{R_m} \left( \int_{|y|\ge R_m/2} \frac{|u_m(y)|^2}{|y|} \d y \right)^{-1} \\
 & \le \sup_{|x|\ge r_m} |xV(x)|  + C_V m^{-1/2},
 \end{align*}
 which is equivalent to \eqref{eq:am<=Z}. Moreover, using
\bq \label{eq:V-eta-Rm2}
\int_{|y|\ge R_m/2} \frac{|u_m(y)|^2}{|y|} \d y \le \frac{2}{R_m} \int_{\R^3}  |u_m(y)|^2 \d y = \frac{2m}{R_m}   
\eq
we also deduce \eqref{eq:local-rm-am} immediately from \eqref{eq:a<Z-1}. 
\end{proof}

The previous estimates provide sufficient tools to prove the nonexistence for short-range potentials, which will be done in the next section. We end this section by extracting a consequence of Lemma  \ref{lem:rm-am}, which will be useful when we deal with Coulomb-type potentials.  

\begin{lemma} \label{lem:Vuu<<m} When $m\ge C_V^3$, we have 
\begin{equation} \label{eq:Vuu<<m}
 \int_{\R^3} \abs{V(x)} |u_m(x)|^2 \d x \le C_{V}^3. 
\end{equation}

\end{lemma}

\begin{proof} Let $r_m$ be as in Lemma \ref{lem:rm-am}. Then from \eqref{eq:am<=Z}, we have 
\[ a_m:=\int_{\R^3} |\chi_{r_m}(x)u_m(x)|^2 \d x \le C_V. \]
Thus, by Lemma~\ref{lem:basic-energy},
\[
\cE_V(\chi_{r_m} u_m) \ge - C_V^2 a_m \ge - C_V^3.
\]
Combining with the lower bound on the Thomas-Fermi and Dirac terms \eqref{eq:complete_square}, we get
$$
\int_{\R^3} |V(x)| |\chi_{r_m}(x)u_m(x)|^2 \d x \le - \cE_V(\chi_{r_m} u_m) + \frac{c_{\rm D}^2}{4 c_{\rm TF}} a_m \le C_V^3 + C_V.
$$
On the other hand, from \eqref{eq:am-V} and \eqref{eq:V-eta-Rm2} we get
$$
\int_{\R^3} |V(x)| |\eta_{r_m}(x) u_m(x)|^2 \d x \le C_V \left( \frac{C_V}{r_m} + \frac{2m}{R_m} \right) \le C_V^2. 
$$
Summing the last two estimates, we find the desired bound.
\end{proof}

\section{Short-range potentials: Proof of Theorem \ref{thm:short-range}} \label{sec:short-range} 

In this section, we demonstrate our general strategy in the simpler case of short-range potentials.  

\begin{proof}[Proof of Theorem \ref{thm:short-range}] The case $V\equiv 0$ has been settled in Lemma \ref{lem:nonexistence-I0}. Now let us consider the case $V\le 0$ and $V\not\equiv 0$. We will assume that $I_V(m_k)$ has a minimizer $u_{m_k}$ with $m_k\to \infty$ and then derive a contradiction. In order to keep  notations simple, we will write $m$ and $u_m$ instead of $m_k$ and $u_{m_k}$. 

Choose $r_m$ as in Lemma \ref{lem:rm-am} and denote
$$ a_m:= \int_{\R^3} |\chi_{r_m}(x)u_m(x)|^2 \d x.$$
Under the short-range assumption $|xV(x)|\to 0$ as $|x|\to \infty$, using $R_m\ge m/C$ in Lemma \ref{lem:Rm>m}, we deduce from \eqref{eq:am<=Z} and  \eqref{eq:local-rm-am}
that $\lim_{m\to \infty} a_m= 0$ and
\begin{align*}
&\liminf_{m\to \infty} \left( I_V(m)-I_V(a_m)-I_0(m-a_m) \right) \\
\ge & \liminf_{m\to \infty} \Big( \cE_V(u_m) - \cE_V(\chi_{r_m} u_m) - \cE_0 (\eta_{r_m} u_m ) \Big)  \ge 0.
\end{align*}
Since $I_V(a_m)\ge -C_V a_m \to 0$ and $I_0(m-a_m)\ge I_0(m) \ge I_V(m)$, we find that
\bq \label{eq:contra-1}
\lim_{m\to \infty} \left( I_V(m)-I_0(m) \right) = 0.
\eq

On the other hand, by Lemma \ref{lem:nonexistence-I0}, there exists a constant $M_0>0$ such that $I_0(m)$ has no minimizers for all $m\ge M_0$. Define
\bq \label{eq:contra-3}
\delta_V:=\inf_{a\in [M_0/2,M_0]} (I_0(a)-I_V(a)).
\eq 
Let us show that $\delta_V>0$. Since the function $a\mapsto I_0(a) - I_V(a)$ is continuous by Lemma~\ref{lem:basic-energy}, it suffices to show that $I_0(a)-I_V(a)>0$ when $a>0$. Indeed, by Lemma \ref{lem:mass-decomp-I0} (ii) and the binding inequality \eqref{eq:non-strict-bind-ineq}, we can write
$$
I_0(a)= I_0(b) + I_0(a-b)
$$
for some $0<b\le a$ such that $I_0(b)$ has a minimizer $v$. Since $\cE_0(u)$ is translation-invariant, $v(\dotv+y)$ is also a minimizer for $I_0(b)$ for all $y\in \R^3$. Since $V\le 0$ and $V \not\equiv 0$, there exists $y\in \R^3$ such that 
$$
I_V(b)-I_0(b) \le \cE_V(v(\dotv +y)) - \cE_0(v(.+y)) = \int_{\R^3} V(x) |v(x+y)|^2 \d x <0.
$$  
By the binding inequality \eqref{eq:non-strict-bind-ineq} again,
$$
I_V(a)\le I_V(b)+I_0(a-b) < I_0(b)+ I_0(a-b)=I_0(a).
$$
Thus $I_0(a)-I_V(a)>0$ for all $a>0$, and hence $\delta_V>0$.

Finally, we show that when $m\ge M_0$, then $I_0(m)-I_V(m)\ge \delta_V$. Indeed, by Lemma \ref{lem:mass-decomp-I0} (ii) and the binding inequality \eqref{eq:non-strict-bind-ineq}, we can write
$$
I_0(m)= I_0(m') + I_0(m-m')
$$
for some $m'\in [M_0/2,M_0]$. Using the binding inequality \eqref{eq:non-strict-bind-ineq} and the definition of $\delta_V$ in \eqref{eq:contra-3}, we find that
$$
I_0(m)-I_V(m) =  I_0(m') + I_0(m-m') - I_V(m) \ge I_0(m')-I_V(m') \ge \delta_V.
$$ 
This is a contradiction because $I_0(m)-I_V(m)\to 0$ as $m\to \infty$ by \eqref{eq:contra-1} and $\delta_V>0$ is independent of $m$.
\end{proof}

For the Coulomb potential in \eqref{eq:def-V}, $|x|V(x)$ does not vanish at infinity and the simple bound $R_m\ge m/C$ in Lemma \ref{lem:Rm>m} is not enough to control the error term in \eqref{eq:local-rm-am}. We need the stronger estimate $R_m \gg m$, which will be derived in the next section. 

\section{Improved radius estimate} \label{sec:improved}

We will use again the notations of Section \ref{sec:radius}. The following is the key estimate to treat Coulomb potentials. 

\begin{lemma}[Improved radius estimate] \label{lem:Rm>m3/2} When $m\ge C_V^3$ we have $ R_m \ge m^{2}/C_V^3$.
\end{lemma}
 
 \begin{proof} We will use successive localization estimates. In the following, $\eps>0$ can be taken as any small, fixed constant (we will eventually choose $\eps=1/4$).

  \proofstep (Improved localization in annulus)
  We show that
  $$ \int_{\Omega} |u_m(x)|^2 \d x \ge  (1-\eps) m $$
where
$$
\Omega=\{x \in \R^3 | R_m - K_m \le |x| \le R_m+ K_m\} \quad \text{for some \,\,} K_m \le \frac{C R_m}{\eps m}.
$$
Let us denote
$$ K^-= \min \left\{K \in \mathbb{N}: \int_{|x| \le R_m-K} |u_m(x)|^2 \d x \le  \frac{\eps}{2} m \right\}$$
and
$$ K^+ = \min \left\{K \in \mathbb{N}: \int_{|x| \ge R_m+K} |u_m(x)|^2 \d x \le \frac{\eps}{2} m \right\}.$$
Then from Lemma~\ref{lem:>2Rm}, we know that, as soon as $\eps m > C_V$, 
$$K^- \le R_m/2, \quad K^+ \le R_m+2.$$
Applying the basic localization estimate \eqref{eq:annulus-0} with $R = R_m +j$ and summing over $j= -K^-, \dots, K^+$ leads to
\begin{multline} \label{eq:sumRm-K}
\sum_{j = -K^-}^{ K^+} \frac{1}{3} \left( \int_{\R^3} |\chi_{R_m+j}(x)u_m(x)|^2 \d x \right) \left( \int_{\R^3} \frac{|\eta_{R_m+j} (y)u_m(y)|^2}{|y|}   \d y \right)   \\
\le - \sum_{j = -K^-}^{ K^+} \int_{\R^3} V(x) |\eta_{R_m+j}(x) u_m(x)|^2 \d x + C m .
\end{multline}
Since $R_m+j \geq R_m/2$ in the range we are considering, we have, similarly to \eqref{eq:Rm>=m-3}, 
\bq \label{eq:sumRm-K-a}
- \int_{\R^3} V(x) |\eta_{R_{m+j}}(x) u_m(x)|^2 \d x  \le \frac{ C_V m}{R_m} 
\eq
Moreover, when $-K^+ +1 \le j \le -1$, we have
$$
\int_{\R^3} |\chi_{R_m+j}(x)u_m(x)|^2 \d x \ge \int_{|x|\le R-K^- +1} |u_m(x)|^2 \d x \ge \frac{\eps}{2} m, 
$$
by the definition of $K^-$, and 
$$
\int_{\R^3} \frac{|\eta_{R_m+j} (y)u_m(y)|^2}{|y|}   \d y \ge \int_{\R^3} \frac{|\eta_{R_m} (y)u_m(y)|^2}{|y|} \d y \ge \frac{m- C_V}{4 R_m+4},
$$
by \eqref{eq:Rm>=m-2}. Therefore, when $m$ and $R_m$ are large,
\bq \label{eq:sumRm-K-b}
\left( \int_{\R^3} |\chi_{R_m+j}(x)u_m(x)|^2 \d x \right) \left( \int_{\R^3} \frac{|\eta_{R_m+j} (y)u_m(y)|^2}{|y|}   \d y \right) 
\ge \frac{\eps m^2}{C R_m}.
\eq
when $-K^- +1 \le j \le -1$. By the same argument, we can prove that \eqref{eq:sumRm-K-b} also holds when  $0 \le j \le K^+ -1$. Substituting \eqref{eq:sumRm-K-a} and \eqref{eq:sumRm-K-b} into \eqref{eq:sumRm-K}, we find that
$$ (K^+ + K^- -1) \frac{\eps m^2}{CR_m} \le (K^+ + K^-+1) \frac{C_V m}{R_m} +  C m.$$
We obtain
$$
K^+ + K^- - 1 \le \frac{C R_m }{\eps m} . 
$$
Thus we can choose $K_m = \max\{ K^-, K^+\}$.

%

\medskip

\proofstep (Further localization in a slab)
We now show that either $R_m\ge m^2/C_V^3$, or there exists a unit vector ${\bf n} \in \sS^2$ such that
\[
\int_{\Omega_{\bf n}} \abs{u_m(x)}^2 \d x \ge (1- 2 \eps)m, 
\]
where 
\[
\Omega_{\bf n} = \{x \in \Omega |-L_m \leq {\bf n} \cdot x \leq L_m \} \quad \text{with some \, \,} L_m \le \frac{ R_m}{\eps (1- \eps) m }.
\]

Let us fix ${\bf n} \in \sS^2$ such that
\[
\int_{\left\{\substack{x \in \Omega \\ {\bf n} \cdot x \ge 0} \right\}} \abs{u_m(x)}^2 \d x 
             = \int_{\left\{\substack{x \in \Omega \\ {\bf n} \cdot x \le 0}\right\}} \abs{u_m(x)}^2 \d x.
\]
This is always possible, since the function
\[
{\bf n} \mapsto \int_{\left\{\substack{x \in \Omega \\ {\bf n} \cdot x \ge 0} \right\}} \abs{u_m(x)}^2 \d x 
       - \int_{\left\{\substack{x \in \Omega \\ {\bf n} \cdot x \le 0}\right\}}\abs{u_m(x)}^2 \d x
\]
is continuous on $\sS^2$ and it integrates to zero.
Now we will use the basic localization estimate \eqref{eq:loc} applied to {\blue half-spaces} perpendicular to $\bf n$.
Take the partition of unity $f^2+g^2 =1$ as in \eqref{eq:def-fg} and set 
$$\varphi^-_\ell (x)= f({\bf n} \cdot x -\ell), \quad \varphi^+_\ell (x) = g ({\bf n} \cdot x -\ell).$$
Also define
$$ L^-= \min \left\{L \in \mathbb{N}: \int_{\Omega} |\varphi^-_{-L}(x) u_m(x)|^2 \d x \le  \frac{\eps}{2} m \right\}$$
and
$$ L^+ = \min \left\{L \in \mathbb{N}: \int_{\Omega} |\varphi^+_L(x) u_m(x)|^2 \d x \le \frac{\eps}{2} m \right\}.$$

We now apply the localization estimate \eqref{eq:loc} with $\chi = \varphi^-_\ell$, $\eta = \varphi^+_\ell$ and sum over $\ell = -L^-, \dots , L^+$. The potential energy can be controlled by Lemma~\ref{lem:Vuu<<m}. This gives
\bq \label{eq:ell-L-L}
\sum_{\ell= L_-}^{L^+}2  D(| \varphi^-_\ell  u_m|^2, |\varphi^+_\ell u_m|^2)   \le 
(L^-+L^+ +1)C_{V}^3  + C m.
\eq
In order to bound the right-hand side, observe that
\begin{align*}
2 D(| \varphi^-_\ell  u_m|^2, & |\varphi^+_\ell u_m|^2) 
         \ge \iint_{\Omega \times \Omega} \frac{|\varphi^-_\ell (x) u_m(x)|^2 |\varphi^+_\ell (y)u_m(y)|^2}{|x-y|} \d x \d y \\
         & \ge \frac{1}{4(R_m+1)} \int_\Omega |\varphi^-_\ell (x) u_m(x)|^2 \d x \int_\Omega |\varphi^+_\ell (y)u_m(y)|^2 \d y,
\end{align*}
where in the last estimate we have used 
$$|x-y|\le |x|+|y|\le 2(R_m+K_m)\le 4(R_m+1)$$
for $x,y\in \Omega$. Now, for $-L^- + 1 \le \ell \le 0$, we have 
\[
\int_\Omega |\varphi^-_\ell(x) u_m(x)|^2 \d x \ge \frac{\eps m}{2}
\]
by the definition of $L^-$ and 
\[
\int_\Omega |\varphi^+_\ell (y)u_m(y)|^2 \d y  \ge \int_{\left\{\substack{y \in \Omega \\ n \cdot y \ge 0}\right\}} \abs{u_m(y)}^2 \d y = \frac{1}{2} \int_\Omega |u_m|^2 \ge (1-\eps) \frac{m}{2}
\]
by the choice of ${\bf n}$. Therefore, 
\[
2 D ( |\varphi^-_\ell  u_m|^2, |\varphi^+_\ell u_m|^2 )\ge \frac{\eps(1-\eps) m^2 }{16(R_m+1)}
\]
when $-L^- + 1 \le \ell \le 0$. This estimate also holds $1 \le \ell \le L^+-1$, by applying the same argument exchanging the roles of both integrals. Thus \eqref{eq:ell-L-L} reduces to 
\[
(L^-+L^+ - 1)\frac{\eps(1-\eps) m^2 }{16 (R_m+1)} \le (L^-+L^+ +1)C_{V}^3  + C m,
\]
and hence we have either $R_m\ge m^2 (C_V^3 \eps (1-\eps))^{-1}$, or 
\[
L^+ + L^- \le \frac{ R_m}{\eps (1- \eps) m }.
\]
In the latter case, we can simply choose $L_m = \max\{L^+, L^-\}$.

\medskip

\proofstep (Final localization in a perpendicular slab) We show that either $R_m\ge m^2/(C_V^3 \eps (1-2 \eps))$, or there exist a unit vector ${\bf v} \in \sS^2$ such that ${\bf v} \cdot {\bf n} = 0$ and 
\[
\int_{\Omega_{{\bf n}, {\bf v}}} \abs{u_m(x)}^2 \d x \ge (1- 3 \eps)m, 
\]
where 
\[
\Omega_{{\bf n}, {\bf v}} = \{x \in \Omega_{\bf n} |-M_m \leq {\bf v} \cdot x \leq M_m \} \quad \text{for some \,\,} M_m \le \frac{CR_m}{\eps (1- 2\eps) m}.
\]

The proof proceeds as before upon replacing $\Omega$ by $\Omega_{\bf n}$. For the first step, note that we may find ${\bf v}\in \sS^2$ such that ${\bf v} \cdot {\bf n}=0$ and
\[
\int_{\left\{ \substack{x \in \Omega_{\bf n} \\ {\bf v} \cdot x \ge 0} \right\}} \abs{u_m(x)}^2 \d x = \int_{\left\{ \substack{x \in \Omega_{\bf n} \\ {\bf v} \cdot x \le 0} \right\}} \abs{u_m(x)}^2 \d x
\]
by following the same continuity argument on the connected set 
$$\{{\bf v} \in \sS^2 | {\bf v} \cdot {\bf n} = 0\} \cong \sS^1.$$

\proofstep(Conclusion) If $R_m\ge m^2/C_V^3$, then we are done. Otherwise, by choosing $\eps = 1/4$, we conclude from the previous steps that 
\bq \label{eq:Omega-m4}
\int_{\Omega_{{\bf n}, {\bf v}}} |u_m(x)|^2 \d x \ge  \frac{m}{4}
\eq
for some ${\bf n}, {\bf v}\in \sS^2$ with ${\bf n} \cdot {\bf v} =0$ and 
$$\Omega_{{\bf n}, {\bf v}}=\{x\in \R^3| R_m -K_m \le |x| \le R_m + K_m, |{\bf n} \cdot x| \le L_m, |{\bf v} \cdot x| \le M_m\},$$
as shown in figure~\ref{fig:1}, with
\bq \label{eq:KLM}
\max\{K_m, L_m, M_m\} \le \frac{CR_m}{m}.
\eq

In order to describe $\Omega_{{\bf n}, {\bf v}}$ more easily, we can choose ${\bf w}\in \sS^2$ such that $({\bf n}, {\bf v}, {\bf w})$ forms an orthonormal basis for the Euclidean space $\R^3$. Using coordinates $(x_1,x_2,x_3) \cong x_1 {\bf n}+ x_2 {\bf v} + x_3 {\bf w}$ in this basis, we can write
\begin{align*}
\Omega_{{\bf n}, {\bf v}} = \Big\{ (x_1,x_2,x_3) \in \R^3 \,\Big|\,  R_m -K_m \le \sqrt{x_1^2+x_2^2+x_3^2} \le R_m + K_m, \\
 |x_1| \le L_m,\,\, |x_2| \le M_m \Big\}.
\end{align*}

   \begin{figure}[]
\begin{tikzpicture}[
eje/.style = {-latex, thin },
measureline/.style = {latex-latex},
front/.style = { thick,dark-gray},
hind/.style = { thick, light-gray}
]

\shade[ball color=gray!10!white,opacity=0.20] (0,0) circle (4cm);
\shade[ball color=white!10!white,opacity=0.25] (0,0) circle (3cm);

\begin{scope}[cm={0.91,0,0,1,(0,0)}] 
 
\draw[hind] (0.4+2.71,0.3-0.8)--(0.4+3.79,0.3-0.8);
\draw[hind] (0.4-2.71,0.3-0.8)--(0.4-3.79,0.3-0.8);

\begin{scope}
 \clip (-4.2,-0.2-0.3) -- (-4.2,1.4-0.3) -- (5,1.4-0.3) -- (5,-0.2-0.3) -- cycle; 
\draw[hind] (0.4,0.3) circle (2.82cm);
\draw[hind] (0.4,0.3) circle (3.87cm);
\end{scope}

\begin{scope}
 \clip (-5,-0.8-0.3) -- (-5,0.8-0.3) -- (4.2,0.8-0.3) -- (4.2,-0.8-0.3) -- cycle; 
\draw[front] (-0.4,-0.3) circle (2.82cm);
\draw[front] (-0.4,-0.3) circle (3.87cm);
\end{scope}

\begin{scope}
  \clip (-5,-0.8-0.3) -- (-4.2,-0.2-0.3) -- (5,-0.2-0.3) -- (4.2,-0.8-0.3) -- cycle; 
  \begin{scope}[cm={1,0,0.4,0.3,(0,-0.5-0.3)}]
  \draw[hind] (0,0) circle (2.89cm ); 
    \draw[hind] (0,0) circle (3.92cm ); 
  \end{scope}
\end{scope}

\begin{scope}
   \clip (-5, 0.8-0.3) -- (-4.2,1.4-0.3) -- (5,1.4-0.3) -- (4.2,0.8-0.3) -- cycle; 
  \begin{scope}[cm={1,0,0.4,0.3,(0,1.1-0.3)}]
  \draw[front] (0,0) circle (2.89cm ); 
    \draw[front] (0,0) circle (3.92cm ); 
  \end{scope}
\end{scope}

\draw[front] (-0.4+2.71,-0.3+0.8)--(-0.4+3.79,-0.3+0.8);
\draw[front] (-0.4+2.71,-0.3-0.8)--(-0.4+3.79,-0.3-0.8);
\draw[front] (-0.4-2.71,-0.3+0.8)--(-0.4-3.79,-0.3+0.8);
\draw[front] (-0.4-2.71,-0.3-0.8)--(-0.4-3.79,-0.3-0.8);

\draw[front] (0.4+2.71,0.3+0.8)--(0.4+3.79,0.3+0.8);

\draw[front] (0.4-2.71,0.3+0.8)--(0.4-3.79,0.3+0.8);

\end{scope}

\shade[fill=blue,opacity=0.15] (-5,-0.8-0.3) -- (-5,0.8-0.3) -- (4.2,0.8-0.3) -- (4.2,-0.8-0.3) -- cycle; 
\shade[fill=blue,opacity=0.15] (-4.2,-0.2-0.3) -- (-4.2,1.4-0.3) -- (5,1.4-0.3) -- (5,-0.2-0.3) -- cycle; 

\shade[fill=blue,opacity=0.10] (-5,0.8-0.3) -- (-4.2,1.4-0.3) -- (5,1.4-0.3) -- (4.2,0.8-0.3) -- cycle;   
\shade[fill=blue,opacity=0.10] (-5,-0.8-0.3) -- (-4.2,-0.2-0.3) -- (5,-0.2-0.3) -- (4.2,-0.8-0.3) -- cycle;  

\filldraw (0,0) circle (1pt);
\draw (0,0) node [below] {\footnotesize $0$};

\draw[eje] (0,0) -- (2,0); \draw (0.65,0) node [below] {\footnotesize $\bf w$};
\draw[eje] (0,0)--(0,2); \draw (0,0.35) node [left] {\footnotesize $\bf n$};
\draw[eje] (0,0)--(0.8,0.6); \draw (0.8,0.6) node [above] {\footnotesize $\bf v$};

\draw[-latex](0,0)-- (1,3.87) node [right] {\footnotesize $R_m+K_m$};
\draw[measureline] (0,3)--(0,4); \draw (0,3.5) node [left] {\footnotesize $2K_m$};
\draw[measureline] (-5,-0.8-0.3) -- (-5,0.8-0.3); \draw (-5,-0.3) node [left] {\footnotesize $2L_m$};
\draw[measureline] (-5,0.8-0.3) -- (-4.2,1.4-0.3); \draw (-5,1.1-0.3) node [above] {\footnotesize $2M_m$};

\draw (-3.3,0) node {\footnotesize{$\Omega_{{\bf n}, {\bf v}}^-$}};
\draw (3.5,0) node {\footnotesize{$\Omega_{{\bf n}, {\bf v}}^+$}};
\end{tikzpicture}
\caption{An illustration of $\Omega_{{\bf n}, {\bf v}}=\Omega_{{\bf n}, {\bf v}}^-\cup \Omega_{{\bf n}, {\bf v}}^+$} \label{fig:1}
\end{figure}

Note that if $(x_1,x_2,x_3)\in \Omega_{{\bf n}, {\bf v}}$, then 
$$
\sqrt{(R_m - K_m)^2 - L_m^2 - M_m^2} \le |x_3| \le \sqrt{(R_m + K_m)^2+ L_m^2 + M_m^2}.
$$ 
Therefore, the diameters of each of the sets 
$$
\Omega_{{\bf n}, {\bf v}}^+ = \{ (x_1,x_2,x_3) \in \Omega_{{\bf n}, {\bf v}} | x_3 \ge 0\}, \quad \Omega_{{\bf n}, {\bf v}}^- = \{ (x_1,x_2,x_3) \in \Omega_{{\bf n}, {\bf v}} | x_3< 0\}
$$
are not larger than
\begin{align*}
&2L_m + 2M_m + \sqrt{(R_m + K_m)^2+ L_m^2 + M_m^2} - \sqrt{(R_m - K_m)^2 - L_m^2 - M_m^2} \\
&  \quad \le  2L_m + 2M_m +  \frac{4R_mK_m + 2K_m^2+2L_m^2+2M_m^2}{R_m} \le \frac{CR_m}{m}.
\end{align*}
Here we have used \eqref{eq:KLM} and the elementary inequality
$$\sqrt{a}-\sqrt{b}=\frac{a-b}{\sqrt{a}+\sqrt{b}} \le \frac{a-b}{ \sqrt{a}}, \quad \forall a\ge b>0.$$
Therefore, 
 $$
 \iint_{\Omega_{{\bf n}, {\bf v}}^+  \times \Omega_{{\bf n}, {\bf v}}^+ }  \frac{|u_m(x)|^2 |u_m(y)|^2}{|x-y|} \d x \d y \ge \frac{m}{CR_m}  \left(\int_{\Omega_{{\bf n}, {\bf v}}^+} |u_m|^2 \right)^2 
 $$
and the same holds with $\Omega_{{\bf n}, {\bf v}}^+$ replaced by $\Omega_{{\bf n}, {\bf v}}^-$. Summing these estimates, then using $a^2+b^2\ge (a+b)^2/2$ for $a,b\ge 0$ and \eqref{eq:Omega-m4}, we find that 
\begin{align*}
D(|u_m|^2 , |u_m|^2)  & \ge \frac{m}{CR_m} \left[ \left(\int_{\Omega_{{\bf n}, {\bf v}}^+} |u_m|^2 \right)^2 + \left(\int_{\Omega_{{\bf n}, {\bf v}}^-} |u_m|^2 \right)^2 \right] \\
& \ge \frac{m}{2CR_m} \left(\int_{\Omega_{{\bf n}, {\bf v}}^+} |u_m|^2 + \int_{\Omega_{{\bf n}, {\bf v}}^-} |u_m|^2  \right)^2 \\
& = \frac{m}{2CR_m} \left(\int_{\Omega_{{\bf n}, {\bf v}}} |u_m|^2  \right)^2 \ge \frac{m^3}{32 C R_m}.
\end{align*}

On the other hand, from $\cE_V(u_m)=I_V(m)<0$ and \eqref{eq:kinetic-estimate} it follows that
\begin{align*}
D(|u_m|^2 , |u_m|^2) \le C_V^2 m.
\end{align*}
Comparing the latter two estimates, we conclude that $R_m\ge m^2/C_V^2$. Thus in all cases, we have $R_m\ge m^2/C_V^3$, which finishes the proof. 
 \end{proof}

\section{Small charges: Proof of Theorem \ref{thm:small-Z} }\label{sec:proof-small-Z} 

We are now ready to prove our main result. 
\begin{proof}[Proof of Theorem \ref{thm:small-Z}] Fix $V$ as in \eqref{eq:def-V} with
$$Z:=\sum_{j=1}^J Z_j \le Z_{\rm c},$$
where $Z_{\rm c}>0$ is a fixed constant.
Let us assume that $I_V(m)$ has a minimizer $u_m$ with $m$ sufficiently large. 
 We will use the notations from Section \ref{sec:radius}. Since we can choose $C_V \le C (Z+1) \le C(Z_{\rm c}+1)$, the $V$-dependence in the constant $C_V$ can be ignored. Recall that $C$ is a generic (large) constant independent of $Z$ and $m$.

Let $r_m$ be as in Lemma \ref{lem:rm-am}. Since $R_m\ge m^2/C$ by Lemma \ref{lem:Rm>m3/2}, we get 
$$r_m\ge \frac{R_m}{m^{1/2}} \ge \frac{m^{3/2}}{C}.$$
For $|x|\ge r_m$, by the triangle inequality we find that
\begin{align*}
\left| |xV(x)|- Z\right|
\le \sum_{j=1}^J Z_j \left| \frac{|{\bf r}_j|}{|x|-|{\bf r}_j|}  \right|   
\le \frac{CZ}{m^{3/2}} .
\end{align*}
Therefore,  from Lemma \ref{lem:rm-am}, Lemma \ref{lem:Rm>m3/2} and $Z \le Z_{\rm c}$ we deduce that
\bq \label{eq:main-proof-am}
a_m:=\int_{\R^3} |\chi_{r_m}u_m|^2 \le Z+ C (Z+1)m^{-1/2} \le Z + C m^{-1/2}
\eq
and 
\begin{align} \label{eq:main-proof-IV-IV-I0-lower}
I_V(m)-I_V(a_m)-I_0(m-a_m) &\ge  \cE_V(u_m)- \cE_V(\chi_{r_m} u_m) - \cE_0(\eta_{r_m} u_m ) \nn\\
&\ge - \frac{C}{m}.
\end{align}

Now let us find an upper bound on $I_V(m)-I_V(a_m)-I_0(m-a_m)$. Recall that by Lemma \ref{lem:nonexistence-I0}, there exist a constant $M_0>0$ such that $I_0(m)$ has no minimizers for all $m\ge M_0$. Note that $m-a_m \ge M_0$ when $m$ is sufficiently large, because $a_m$ is bounded by \eqref{eq:main-proof-am}. Therefore, by Lemma \ref{lem:mass-decomp-I0} and the binding inequality \eqref{eq:non-strict-bind-ineq}, we can decompose
\bq \label{eq:def-bm}
I_0(m-a_m)= I_0(b_m) + I_0(m-a_m-b_m)
\eq
for some $b_m\in [M_0/2,M_0]$. By \eqref{eq:non-strict-bind-ineq} again,
$$I_V(m)\le I_V(b_m) + I_0(a_m)+ I_0(m-a_m-b_m).$$
Therefore, 
\bq \label{eq:main-proof-IV-IV-I0-upper}
 I_V(m)-I_V(a_m)-I_0(m-a_m) \le I_0(a_m)-I_V(a_m) + I_V(b_m) - I_0(b_m) .
\eq
It remains to estimate $I_0(a_m)-I_V(a_m)$ and $I_0(b_m)-I_V(b_m)$ separately. 

First, we consider $I_0(a_m)-I_V(a_m)$. By the definition of $I_V(a_m)$, for every fixed $m$ and $\eps>0$, we can find $v\in H^1(\R^3)$ such that 
$$\int |v|^2=a_m\quad \text{and} \quad I_V(a_m) +\eps \ge \cE_V(v).$$
From $I_V(a_m)<0$ and Lemma \ref{lem:basic-energy}, we find that 
$$\int |\nabla v|^2 \le C(a_m+\eps). $$
Using the variational principle and Hardy's inequality, we can estimate 
\begin{align*}
I_0(a_m)- I_V(a_m)-\eps &\le \cE_0(v) - \cE_V(v)  =  - \int V |v|^2 
= \sum_{j=1}^J Z_j \int \frac{|v(x)|^2}{|x-{\bf r}_j|} \d x  \\
&\le  \sum_{j=1}^J Z_j  \left( \int |v(x)|^2 \d x \right)^{1/2} \left( \int \frac{|v(x)|^2}{|x-{\bf r}_j|^2} \d x \right)^{1/2} \\
& \le  \sum_{j=1}^J Z_j  \left( \int |v(x)|^2 \d x \right)^{1/2} \left( 4\int |\nabla v(x)|^2 \d x\right)^{1/2} \\
& \le CZ \sqrt{a_m(a_m+\eps)}.
\end{align*}
Since $\eps>0$ can be taken arbitrarily small, we conclude that
\begin{align} \label{eq:I0am-IVam}
I_0(a_m)- I_V(a_m) \le CZa_m. 
\end{align}

Now, we consider $I_0(b_m)-I_V(b_m)$. We have
\bq \label{eq:IVa-I0b}
I_0(b_m)-I_V(b_m) \ge \lambda_0 Z, \quad \forall Z>0,
\eq
where $\lambda_0$ is defined by
$$
\lambda_0 := \inf \left\{ \frac{I_0(b)-I_V(b)}{Z} : Z=\sum_{j=1}^J Z_j >0, M_0 \ge b \ge \frac{M_0}{2} \right\}. 
$$
Here, the minimization is over potentials $V$ of the form \eqref{eq:def-V} with a fixed number of nuclei $J \in \N$ 
and fixed nuclear positions $\{{\bf  r}_j\}$, but the nuclear charges $\{Z_j\}$ and the mass of the electrons $b$ are allowed to vary. We will show that $\lambda_0>0$ at the end of the proof.

Assuming this for the moment, by inserting \eqref{eq:I0am-IVam} and \eqref{eq:IVa-I0b} into \eqref{eq:main-proof-IV-IV-I0-upper}, we find that
$$
I_V(m)-I_V(a_m)-I_0(m-a_m)  \le CZa_m - \lambda_0 Z. 
$$
Combining the latter estimate with \eqref{eq:main-proof-am} and \eqref{eq:main-proof-IV-IV-I0-lower}, we arrive at
\begin{align*}
- \frac{C}{m} &\le I_V(m)-I_V(a_m)-I_0(m-a_m)  \le CZa_m - \lambda_0 Z \\
& \le CZ \Big( Z+ C m^{-1/2} \Big) - \lambda_0 Z.
\end{align*}

when  $Z < \lambda_0 / C$, this gives a contradiction when $m\ge M_{\rm c}$ with $M_{\rm c}$ large. 

\medskip

To finish the proof, we still have to show that $\lambda_0>0$. By the definition of $\lambda_0$, there exist sequences $\{b_k\}_k \subset [M_0/2,M_0]$ and $\{Z_j^{(k)}\}_{j,k}$ such that  
$$ 
\lambda_0= \lim_{k\to \infty} \frac{I_0(b_k)-I_{V_k}(b_k)}{\sum_{j=1}^J Z_j^{(k)}} \quad \text{with }\quad V_k= -\sum_{j=1}^J \frac{Z_j^{(k)}}{|x-{\bf r}_j|}.
$$ 
Recall that by Lemma \ref{lem:mass-decomp-I0} (iii), there exists a constant $m_0>0$ such that $I_0(m)$ has minimizers for all $m \le m_0$. By Lemma  \ref{lem:mass-decomp-I0} (ii) and the binding inequality \eqref{eq:non-strict-bind-ineq}, we can decompose
$$
I_0(b_k)= I_0(b'_k)+ I_0(b_k-b'_k)
$$   
for some $b_k' \in [m_0/2,b_k]$ such that $I_0(b'_k)$ has a minimizer $v_k$. By passing to a subsequence if necessary, we can assume that $b_k' \to b' \in [m_0/2,M_0]$. Since
$$ \lim_{k\to \infty}  \cE_0((b_k'/b')^{1/2}v_k) = \lim_{k\to \infty}  \cE_0(v_k)= \lim_{k\to \infty}  I_0(b_k') = I_0(b'), $$
the sequence $\{(b_k/b')^{1/2}v_k\}$ is a minimizing sequence for $I_0(b')$. By Lemma  \ref{lem:mass-decomp-I0}(i), up to subsequences and translations, we can assume that $(b'_k/b')^{1/2}v_k$, and hence  $v_k$, converges weakly in $H^1(\R^3)$ to some $v\not \equiv 0$. Thus, for every $j=1,2,...,J$,
\bq \label{eq:converge-H1-Coulomb}
\lim_{k\to \infty} \int_{\R^3} \frac{|v_k(x)|^2}{|x-{\bf r}_j|} \d x = \int_{\R^3} \frac{|v(x)|^2}{|x-{\bf r}_j|} \d x >0.
\eq
By the binding inequality \eqref{eq:non-strict-bind-ineq} and the variational principle, we have
\begin{align*}
I_0(b_k)-I_{V_k}(b_k) &= I_0(b'_k)+ I_0(b_k-b'_k) - I_{V_k}(b_k) \\
&\ge I_0(b'_k)- I_{V_k}(b'_k) \ge \cE_0(v_k)-\cE_{V_k}(v_k) =- \int V_k |v_k|^2.
\end{align*}
Therefore, we conclude that
\begin{align*}
\lambda_0 &= \lim_{k\to \infty} \frac{I_0(b_k)-I_{V_k}(b_k)}{\sum_{j=1}^J Z_j^{(k)}} \ge \lim_{k\to \infty} \frac{- \int V_k |v_k|^2}{\sum_{j=1}^J Z_j^{(k)}} \\
&\ge \lim_{k\to \infty} \, \min_{j=1,..,J}  \int_{\R^3} \frac{|v_k(x)|^2}{|x-{\bf r}_j|} \d x = \min_{j=1,..,J}  \int_{\R^3} \frac{|v(x)|^2}{|x-{\bf r}_j|} \d x >0 .
\end{align*}
This completes the proof.
\end{proof}

\section{Partial nonexistence: Proof of Theorem \ref{thm:no-stable-minimizers}} \label{sec:instability}

In this section, we will consider a general Coulomb-type potential
$$ V\in L_{\rm loc}^{3/2}(\R^3) \quad \text{and}\quad \limsup_{|x|\to \infty} |x V(x)| <\infty.$$

To prove Theorem ~\ref{thm:no-stable-minimizers} (i), we need the following analogue of Lemma \ref{lem:mass-decomp-I0} (ii).

\begin{lemma} \label{lem:mass-decomp-IV} For every $m>0$, there exists $m \ge m' > 0$ such that 
$$ I_V(m)=I_V(m')+I_0(m-m')$$
and $I_V(m')$ has a minimizer. 
\end{lemma}

\begin{proof} This is a typical result from the concentration-compactness method \cite{Lions-84,Lions-84b}. For the reader's convenience, we quickly sketch its proof. Let $\{v_n\}$ be a minimizing sequence for $I_V(m)$. After extracting a subsequence, we may assume that $v_n \wto v$ weakly in $H^1(\R^3)$. By the same arguments as those leading to \eqref{eq:converge-H1-Coulomb} above and \eqref{eq:energy-split}, \eqref{eq:energy-split-2} in the appendix, we find that 
\begin{align*}
\lim_{n\to \infty} \Big( \cE_V(v_n)- \cE_V(v) -  \cE_0(v_n-v) \Big)&=0,\\
\lim_{n\to \infty} \int_{\R^3} \Big| | v_n|^2 - |v|^2 - |v_n-v|^2 \Big| &=0  .
\end{align*}
The second one of these equations implies that 
$$\lim_{n\to \infty}\int |v_n-v|^2 = m-m' \quad \text{with} \quad m':=\int |v|^2 \le m.$$
Using the binding inequality \eqref{eq:non-strict-bind-ineq} and the variational principle, we get
\begin{align*}
0 &\ge I_V(m)- I_V(m')- I_0(m-m') \ge I_V(m)- \cE_V(v)- I_0(m-m') \\
&\ge \liminf_{n\to \infty} \Big( \cE_V(v_n) - \cE_V(v) -\cE_0 (v_n-v) \Big) =0.
\end{align*}
Thus we conclude that
$$ I_V(m)=I_V(m') +I_0(m-m')$$
and $\cE_V(v)= I_V(m')$,  namely $v$ is a minimizer for $I_V(m')$. Finally, we have $m'>0$ because $I_0(s)>I_V(s)$ for all $s>0$, which has been shown in the proof of Theorem \ref{thm:short-range} in Section \ref{sec:short-range}. 
\end{proof}

Now we are ready to prove the nonexistence of stable minimizers.

\begin{proof}[Proof of Theorem ~\ref{thm:no-stable-minimizers} (i)] We distinguish two cases.

\medskip

\noindent{\bf Case 1:} Assume that $I_V(m)$ has no miminizers for some $m>0$. Then by Lemma \ref{lem:mass-decomp-IV}, we can decompose
$$ I_V(m)=I_V(a)+I_0(m-a)$$
for some $m>a>0$ such that $I_V(a)$ has a minimizer. By Lemma \ref{lem:mass-decomp-I0} (ii), we can further decompose
$$ I_0(m-a)=I_0(b)+I_0(m-a-b)$$
for some $m-a\ge b>0$ such that $I_0(b)$ has a minimizer. Combining the last two equalities and using the binding inequality \eqref{eq:non-strict-bind-ineq}, we have
\begin{align*}
I_V(m)&=I_V(a)+I_0(m-a)= I_V(a)+I_0(b)+ I_0(m-a-b) \\
&\ge I_V(a+b)+I_0(m-a-b)\ge I_V(m).
\end{align*}
Therefore, we conclude that $I_V(a)+I_0(b)=I_V(a+b)$, as desired.

\medskip

\noindent{\bf Case 2:} Assume that $I_V(m)$ has minimizers for all $m>0$. Define $\widetilde Z = \limsup_{\abs x \to \infty} \abs{x V(x)}$. Then from the proof of Theorem \ref{thm:small-Z} (see equations \eqref{eq:main-proof-am},  \eqref{eq:main-proof-IV-IV-I0-lower} and \eqref{eq:def-bm}), we can find $a_m  \ge 0$ and $b_m \in [m_0/2, M_0]$ such that 
\bq \label{eq:thm-binding-1}
\liminf_{m\to \infty} \Big( I_V(m) - I_V(a_m) - I_0(b_m) - I_0(m-a_m-b_m) \Big) \ge 0.
\eq
Restricting to a subsequence and using \eqref{eq:am<=Z}, we can assume 
$$\lim_{m\to \infty} a_m = a \in [0,\widetilde Z] \quad \text{and}\quad \lim_{m\to\infty} b_m = b \in [M_0/2,M_0].$$
From \eqref{eq:thm-binding-1} and the binding inequality \eqref{eq:non-strict-bind-ineq}, we get
\begin{align*}
0&\le \liminf_{m\to \infty} \Big( I_V(m) - I_0(m-a_m-b_m)  - I_V(a_m) - I_0(b_m) \Big) \\
&\le \liminf_{m\to \infty}  \Big( I_V(a_m+b_m) - I_V(a_m) - I_0(b_m) \Big) \\
& = I_V(a+b)- I_V(a)- I_0(b) \le 0.
\end{align*}
Therefore,
$$I_V(a+b)=I_V(a)+I_0(b).$$
Since $b>0$ and $I_0(s)>I_V(s)$ for all $s>0$ (see the proof of Theorem \ref{thm:short-range}), we find that $a>0$. By the initial assumption in Case 2, we know that $I_V(a)$ has a minimizer. We do not know yet if $I_0(b)$ has a minimizer. However, by Lemma \ref{lem:mass-decomp-I0} (ii) we can always decompose
$$ I_0(b)=I_0(b')+I_0(b-b')$$
for some $b\ge b'>0$ such that $I_0(b')$ has a minimizer. Using the binding inequality \eqref{eq:non-strict-bind-ineq} again, we get
\begin{align*}
I_V(a+b) &= I_V(a)+ I_0(b) = I_V(a)+ I_0(b') + I_0(b-b')  \\
&\ge I_V(a+b') + I_0(b-b') \ge I_V(a+b).
\end{align*}
Thus 
$$
I_V(a+b')= I_V(a)+I_0(b')
$$
with $a>0,b'>0$, and both $I_V(a)$ and $I_0(b')$ have minimizers. 
\end{proof}

Now we turn to the proof of Theorem ~\ref{thm:no-stable-minimizers} (ii). The estimates in Section \ref{sec:radius} are actually sufficient for this part. The key ingredient for the proof is the following  lemma. 

\begin{lemma} \label{lem:ux>=}Assume that $m\ge C_V^3$ and $I_V(m)$ has a minimizer $u_m$. Let $R_m$ be the radius of the system defined in Section \ref{sec:radius}. Then 
$$ \esssup_{|x|\ge R_m/2} |u_m(x)| \ge \frac{1}{C}.$$
\end{lemma}

\begin{proof} First, by the binding inequality \eqref{eq:non-strict-bind-ineq} we find that
$$ I_V(m) \le I_V(m- \lfloor m \rfloor ) + \lfloor m \rfloor I_0(1),$$
where $\lfloor m \rfloor$ is the largest integer that is not bigger than $m$. Since $I_V(m- \lfloor m \rfloor) \le 0$, $I_0(1)<0$ and $\lfloor m \rfloor > m-1$, we obtain
$$ \cE_V(u_m)= I_V(m) \le (m-1) I_0(1) \le -\frac{m}{C}.$$
Combining with the bound on the potential energy in Lemma \ref{lem:Vuu<<m}, we get 
\bq \label{eq:u10-u8}
\int_{\R^3} \Big(c_{\rm TF} |u_m|^{10/3}-c_{\rm D} |u_m|^{8/3} \Big) \le \cE_V(u_m) - \int V|u_m|^2 \le   -\frac{m}{C}
\eq
for $m$ large (say, $m \ge C_V^3$). Moreover, from the pointwise estimate \eqref{eq:complete_square}
and Lemma \eqref{lem:>2Rm}, it follows that
\begin{align*}
\int_{|x| \le R_m/2}   \Big(c_{\rm TF} |u_m(x)|^{10/3}- &c_{\rm D} |u_m(x)|^{8/3} \Big) \d x \\
&\ge  - \frac{c_{\rm TF}^2}{4c_{\rm D}} \int_{|x| \le R_m/2}  |u_m(x)|^2 \d x \ge - C_V.
\end{align*}
Therefore, \eqref{eq:u10-u8} implies that  
$$
\int_{|x| \ge R_m/2} -c_{\rm D} |u_m|^{8/3} \le \int_{|x|\ge R_m/2} \Big(c_{\rm TF} |u_m|^{10/3}-c_{\rm D} |u_m|^{8/3} \Big) \le  -\frac{m}{C}
$$
for $m$ large. On the other hand, if we denote $ 
\lambda:= \esssup_{|x|\ge R_m/2} |u_m(x)|,
$
then 
$$
\frac{m}{C} \le \int_{|x|\ge R_m/2} c_{\rm D} |u_m|^{8/3} \le  c_{\rm D} \lambda^{2/3} \int_{|x|\ge R_m/2} |u_m(x)|^2 \d x \le C \lambda^{2/3}m.
$$
Thus $\lambda \ge 1/C$. 
\end{proof} 


Now we are ready to conclude the proof.

\begin{proof}[Proof of Theorem \ref{thm:no-stable-minimizers} (ii)] Assume that $I_V(m)$ has a radially symmetric minimizer $u_m$ for $m$ large (say, $m\ge C_V^3$). Note that (see \cite[Theorem 7.8]{LieLos-01})
$$\cE_V(u)-\cE_V(|u|)=\int_{\R^3} |\nabla u|^2 - \int_{\R^3} |\nabla |u||^2 \ge 0, \quad \forall u\in H^1(\R^3).$$
Therefore, by replacing $u_m$ by $|u_m|$ if necessary, we can assume that $u_m\ge 0$. Since $u_m$ is radially symmetric, there exists $v_m:[0,\infty)\to [0,\infty]$ such that
$$
u_m(x)= v_m(|x|) \quad \text{for a.e.\,} x\in \R^3.
$$
We have
\bq \label{eq:kinetic-vm-0}
m =  \int_{\R^3} | u_m(x)|^2 \d x = 4\pi \int_{0}^\infty |v_m(r)|^2 r^2 \d r.
\eq
Moreover, from \eqref{eq:kinetic-estimate} and $\cE_V(u_m)=I_V(m)<0$, it follows that 
\bq \label{eq:kinetic-vm-1}
C_V^2 m \ge  \int_{\R^3} |\nabla u_m(x)|^2 \d x = 4\pi \int_{0}^\infty |v_m'(r)|^2 r^2 \d r. 
\eq
In particular, since $v_m\in H^1([a,\infty))$ for all $a>0$, the function $r\mapsto v_m(r)$ is continuous for $r>0$.  From Lemma \eqref{lem:ux>=}, there exists $r_1\ge R_m/2$ such that $v_m(r_1)\ge 1/C$. Since $v_m$ is continuous and it vanishes at infinity, there exists $r_2 > r_1$ such that
\bq  \label{eq:kinetic-vm-2}
v_m(r_2)=\frac{v_m(r_1)}{2} \quad \text{ and } \quad v_m(r) \ge v_m(r_2) \quad \text{for all\,\,} r\in [r_1,r_2]. 
\eq
From \eqref{eq:kinetic-vm-0} and \eqref{eq:kinetic-vm-2}, we obtain
\begin{align} \label{eq:kinetic-vm-3}
\frac{m}{ 4 \pi} &\ge  \int_{r_1}^{r_2} |v_m(r)|^2 r^2 \d r \ge \frac{|v_m(r_1)|^2}{4} \int_{r_1}^{r_2} r^2 \d r  = \frac{|v_m(r_1)|^2}{12} (r_2^3-r_1^3). 
\end{align}
On the other hand, from \eqref{eq:kinetic-vm-1}, \eqref{eq:kinetic-vm-2} and H\"older's inequality,
\begin{align} \label{eq:kinetic-vm-4}
C_V^2 m \ge \int_{r_1}^{r_2} |v_m'(r)|^2 r^2 \d r & \ge \left( \int_{r_1}^{r_2} v_m'(r) \d r \right)^2 \left( \int_{r_1}^{r_2} r^{-2}\d r \right)^{-1} \nn\\
& =  \left( v_m(r_2)- v_m(r_1)\right)^2 \frac{r_1r_2}{r_2-r_1} \nn \\ 
& =\frac{|v_m(r_1)|^2}{4} \cdot \frac{r_1r_2}{r_2-r_1} .
\end{align}
Finally, we multiply \eqref{eq:kinetic-vm-3} with \eqref{eq:kinetic-vm-4}, then use $v_m(r_1)\ge 1/C$ and $r_2>r_1\ge R_m/2 \ge m/C$ by Lemma \ref{lem:Rm>m}. This yields
$$ C_V^2 m^2 \ge \frac{|v_m(r_1)|^4}{48} r_1 r_2 (r_1^2+r_1 r_2 + r_2^2) \ge \frac{m^4}{C}.$$
Thus $m\le C_V$, which finishes the proof. 
\end{proof}

\begin{remark} If we use the improved estimate in Lemma \ref{lem:Rm>m3/2}, then we can conclude faster from \eqref{eq:kinetic-vm-4} and $r_1r_2/(r_2 - r_1) \ge r_1 \ge R_m/2 \ge m^2/C_V^3$.
\end{remark}

\begin{appendix}

\section*{Appendix. Existence for $I_0(m)$}

\begin{proof}[Proof of Lemma \ref{lem:mass-decomp-I0}] (i) Let $\{v_n\}$ be a minimizing sequence for $I_0(m)$. Then $\{v_n\}_{n=1}^\infty$ is bounded in $H^1(\R^3)$. Following \cite{Lions-84}, we define 
\begin{align*}
\mathfrak{M}(\{v_n\}) &= \lim_{R\to \infty} \limsup_{n\to \infty} \sup_{y\in \R^3} \int_{|x-y|\le R}|v_n(x)|^2 \d x,
\end{align*}
which can be understood as the largest chunk of mass which stays in a bounded region (up to subsequences and translations). 
An equivalent characterization is
\begin{align*}
\mathfrak{M}(\{v_n\}) &= \sup \bigl\{ \norm{v}_{L^2} \bigl|
              \exists \{y_k\} \subset \R^3 : v_{n_k}(\dotv + y_k) \rightharpoonup  v \text{ weakly in } H^1(\R^3) \bigr\}.
\end{align*}
From the proof of Lemma I.1 in \cite{Lions-84b} (see also \cite{Lewin-10}), we have
\bq \label{eq:M>=Lp}
 \limsup_{n\to \infty} \int_{\R^3} |v_n|^{10/3} \le C \big(\mathfrak{M}(\{v_n\})\big)^{2/3}    \limsup_{n\to \infty} \| v_n\|_{H^1(\R^3)}^2 .
 \eq
for some universal constant $C>0$ independent of $\{v_n\}$. 
By Lemma~\ref{lem:basic-energy}, $ \| v_n\|_{H^1(\R^3)}^2 \le C m$ .
Combining this bound with $I_0(m) < 0$ for all $m > 0$, we see that
\begin{align*}
0 > I_0 (m)  = \lim_{n \to \infty} \cE_0(v_n)  
      &\ge - c_{\rm D} \limsup_{n \to \infty} \int \abs{v_n}^{8/3} \\
      &\ge  - c_{\rm D} m^{1/2} \limsup_{n \to \infty}  \Bigl(\int \abs{v_n}^{10/3} \Bigr)^{1/2}\\ 
      & \ge -C m \big(\mathfrak{M}(\{v_n\})\big)^{1/3} ,
\end{align*}
and therefore $\mathfrak{M}(\{v_n\}) > 0$.

Now from the characterization of $\mathfrak{M}(\{v_n\})$, by passing to a subsequence if necessary, we can find $\{y_n\}\subset \R^3$ and $v^{(1)} \in H^1(\R^3)$ such that 
$$ v_n(.+y_n) \wto v^{(1)} \text{ weakly in } H^1(\R^3) \quad \text{ and } \quad \int_{\R^3} |v^{(1)}|^2 \ge \frac{1}{2}\mathfrak{M} (\{v_n\}).$$ 
This proves (i).

\medskip

\noindent (ii) Since $\cE_0(u)$ is translation-invariant, the sequence $\{v_n(.+y_n)\}$ is also a minimizing sequence for $I_0(m)$. Therefore, we can assume $y_n=0$ for all $n$. Thus $ v_n = v^{(1)}+ v^{(1)}_n$ with $v^{(1)}_n \wto 0$ weakly in $H^1(\R^3)$. From Sobolev's embedding theorem \cite[Corollary 8.7]{LieLos-01}, by passing to a subsequence if necessary, we have $v_n^{(1)}(x) \to 0$ for a.e. $x\in \R^3$. Let us show that
\bq \label{eq:energy-split}
\lim_{n\to \infty} \Big( \cE_0(v_n)- \cE_0(v^{(1)}) -  \cE_0(v^{(1)}_n) \Big)=0. 
\eq
First, since $\nabla v_n^{(1)} \wto 0$ weakly in $L^2$, we have
\bq \label{eq:energy-split-1}
\lim_{n\to \infty} \int \Big( |\nabla v_n|^2 - |\nabla v^{(1)}|^2 - |\nabla v^{(1)}_n|^2 \Big)  
= \lim_{n\to \infty} 2\Re (\nabla v^{(1)}, \nabla v_n^{(1)}) =0. 
\eq
By Fatou's lemma with remainder term of Brezis and Lieb \cite{BreLie-83}, we get
\bq \label{eq:energy-split-2}
 \lim_{n\to \infty} \int_{\R^3} \Big| | v_n|^p - |v^{(1)}|^p - |v^{(1)}_n|^p \Big| =0 \quad \forall 2\le p \le 6. 
\eq
Combining with the Hardy-Littewood-Sobolev inequality \cite[Theorem 4.3]{LieLos-01}, 
we obtain
\begin{align} \label{eq:energy-split-3}
 \lim_{n\to \infty} D(|v_n|^2, |v_n)|^2)  
=& \lim_{n\to \infty} \Bigl( D(|v^{(1)}|^2, |v^{(1)}|^2) + D( |v^{(1)}_n|^2, |v^{(1)}_n|^2) \Bigr)  \nn \\
& \qquad + \lim_{n\to \infty} 2 D(|v^{(1)}|^2, |v^{(1)}_n|^2) .
\end{align}
But the last term tends to $0$,
since $|v^{(1)}_n|^2*|\dotv|^{-1}$ is bounded in $L^\infty(\R^3)$ and converges to $0$ pointwise. 
From \eqref{eq:energy-split-1},\eqref{eq:energy-split-2}  and \eqref{eq:energy-split-3}, we obtain the decomposition of the energy \eqref{eq:energy-split}.  Consequently, 
$$I_0(m)=\lim_{n\to \infty} \cE_0(v_n)=\cE_0(v^{(1)})+ \lim_{n\to \infty} \cE(v_n^{(1)}) \ge I_0(m_1)+ I_0(m-m_1)$$
where
$$m_1:=\int |v^{(1)}|^2 \ge \frac 1 2 \mathfrak{M}(\{v_n\}).$$
Combining with the binding inequality \eqref{eq:non-strict-bind-ineq}, we conclude that 
$$I_0(m)=I_0(m_1)+I_0(m-m_1).$$
Moreover, $v^{(1)}$ is a minimizer for $I_0(m_1)$ and $\{v_n^{(1)}\}$ is a minimizing sequence for $I_0(m-m_1)$. By repeating the above argument with $v_n$ replaced by the remainder term $v_n^{(1)}$, we find a subsequence and translations such that $v_n^{(1)}=v^{(2)}+ v_n^{(2)}$. Iterating this procedure, we construct sequences $\{v^{(j)}\}_j$ and $\{v^{(j)}_n\}_{j,n}$ such that 
\begin{itemize}
\item  $v^{(j)}$ is a minimizer for $I_0(m_j)$ with $m_j=\int |v_j|^2$ for all $j \ge 1$,

\item $I_0(m) =  \sum_{j=1}^K I_0(m_j) + \lim_{n\to \infty} \cE_0(v^{(K)}_n)$ for all $K \ge 1$,

\item $m=  \sum_{j=1}^K m_j + \lim_{n\to \infty} \int |v^{(K)}_n|^2$,

\item $m_{j+1}  \ge \frac{1}{2}\mathfrak{M} (\{v^{(j)}_n\})$ for all $j \ge 1$. 
\end{itemize}
The last two points imply that $m_j\to 0$ and $\mathfrak{M} (\{v^{(j)}_n\}) \to 0$ as $j\to \infty$. Moreover, $\limsup_{n\to \infty} \|v^{(j)}_n\|_{H^1(\R^3)}$ is bounded uniformly in $j$ because $ \liminf_{n\to \infty} \cE_0(v^{(j)}_n)$ is bounded uniformly in $j$. Therefore we deduce from \eqref{eq:M>=Lp} that
$$ \lim_{j\to \infty} \limsup_{n\to \infty} \|v_n^{(j)}\|_{L^{p}} =0, \quad \text{ for all } 2<p\le 10/3. $$
Consequently 
$$
\liminf_{K\to \infty} \liminf_{n\to \infty}\cE_0(v^{(K)}_n) \ge 0,$$
and hence
$$I_0(m) = \lim_{K\to \infty} \Big( \sum_{j=1}^K I_0(m_j) + \lim_{n\to \infty} \cE_0(v^{(K)}_n) \Big) \ge \sum_{j=1}^\infty I_0(m_j) \ge I_0\Big(\sum_{j=1}^\infty m_j\Big). $$
The last estimate follows from the binding inequality \eqref{eq:non-strict-bind-ineq}. Since $s\mapsto I_0(s)$ is strictly decreasing by Lemma~\ref{lem:binding}, we conclude that 
$$m=\sum_{j=1}^\infty m_j \quad \text{ and } \quad I_0(m) = \sum_{j=1}^\infty I_0(m_j).$$

\noindent (iii) It suffices to prove that when $m$ is sufficiently small,
\bq \label{eq:strict-binding-smallm}
I_0(m)< I_0(m')+ I_0(m-m'), \quad \forall 0<m'< m.
\eq
Then the existence of minimizers for $I_0(m)$ follows from part (ii). The following proof is adapted from the treatment of S\'anchez and Soler \cite{SanSol-04b} for the Schr\"odinger-Poisson-Slater model (which is  the TFDW theory with $c_{\rm TF}= 0$). 

First, note that it suffices to prove \eqref{eq:strict-binding-smallm} when $I_0(m')$ has a minimizer. Otherwise, by part~(ii), we can decompose further $I_0(m')=I_0(m'')+I_0(m'-m'')$ for some $0<m''<m'$ such that $I_0(m'')$ has a minimizer and then prove $I_0(m)< I_0(m'')+I_0(m-m'') \le I_0(m')+ I_0(m-m')$.
 
Let us introduce the short-hand notations
\begin{align*}
A_u &= c_{\rm W} \int  |\nabla u|^2, \quad B_u=c_{\rm TF} \int |u|^{10/3}, \\
C_u &= c_{\rm D} \int |u|^{8/3}, \quad D_u= D(\abs{u}^2, \abs{u}^2)
\end{align*}
such that
$$ \cE_0(u)= A_u + B_u - C_u + D_u.$$ 
Following the argument in \cite[Prop. 2.8]{SanSol-04b}, we have the identity 
\begin{align}
I_0(m) &=  \inf_{\int |u|^2 =m} \mathcal{E}_0(u) = \inf_{\int |u|^2 =1} \mathcal{E}_0(m^2 u(m \dotv )) = \inf_{\int |u|^2=1} \inf_{\ell>0} \mathcal{E}_0(m^2\ell^{3/2}  u( m\ell \dotv)) \nn\\
&= \inf_{\int |u|^2=1} \inf_{\ell>0} \left\{ \ell^2 \Big( m^3 A_u + m^{11/3} B_u \Big) - \ell \Big( m^{7/3} C_u - m^3 D_u \Big) \right\} \nn\\
& = \inf_{\int |u|^2=1} -\frac{m^{5/3}\Big( C_u - m^{2/3} D_u \Big)_+^2 }{ 4\Big(A_u + m^{2/3}B_u\Big) }. \label{eq:I0=frac}
\end{align}
The last equality follows from the variational formula 
$$ \inf_{\ell>0} \Big( \ell^2 a - \ell b \Big) = - \frac{b_+^2}{4a}\quad \text{with}\quad a>0, \, b\in \mathbb{R}, \, b_+=\max\{b,0\}.$$
For every $u\in H^1(\R^3)$ with $\int |u|^2=1$, let us consider the function 
$$
h_u(s)= \frac{s\Big( C_u - s D_u\Big)_+^2 }{A_u + sB_u}.
$$
By H\"older's inequality, Sobolev's inequality and the Hardy-Littlewood-Sobolev inequality, there is a universal constant $C>0$ such that for all $ u\in H^1(\R^3)$,
\bq \label{eq:ABCD}
B_u \le C A_u \Big( \int |u|^2\Big)^{2/3}, \quad D_u \le C C_u \Big( \int |u|^2\Big)^{2/3}.
\eq
Therefore, when $\int |u|^2=1$ and $s>0$ is sufficiently small (independently of $u$), we find that
\begin{align*}
h_u(s) &=  \frac{s\Big( C_u - s D_u\Big)^2 }{ A_u + sB_u}, \\
\frac{dh_u(s)}{d s} &= \frac{( C_u - s D_u ) (A_u C_u - 3 s A_u D_u - 2s^2 B_u D_u) }{(A_u + sB_u)^2}>0.
\end{align*}

Thus $s\mapsto h_u(s)$ is strictly increasing when $s>0$ is small. We are now ready to derive the strict binding inequality \eqref{eq:strict-binding-smallm} for $m$ small. From \eqref{eq:I0=frac} and the monotonicity of $s\mapsto h_s(u)$ we have that
\bq \label{eq:strict-binding-smallm-a}
\frac{I_0(m-m')}{m-m'} =  \inf_{\int |u|^2} - h_u((m-m')^{2/3}) \ge \inf_{\int |u|^2} - h_u(m^{2/3}) = \frac{I_0(m)}{m} . 
\eq
Moreover, since $I_0(m')$ has a minimizer, we can strengthen \eqref{eq:I0=frac} to
$$
\frac{I_0(m')}{m'} = \inf_{\int |u|^2=1} - h_u((m')^{2/3}) = -  h_{v}((m')^{2/3})  
$$
for $v\in H^1(\R^3)$ a suitably rescaled minimizer of $I_0(m')$ with $\int |v|^2=1$. Therefore, using the strict monotonicity of $s\mapsto h_{v}(s)$ we get
\bq \label{eq:strict-binding-smallm-b}
\frac{I_0(m')}{m'} = -  h_{v}((m')^{2/3}) >  -  h_{v}((m)^{2/3}) \ge  \inf_{\int |u|^2} - h_u(m^{2/3}) = \frac{I_0(m)}{m}. 
\eq
Combining \eqref{eq:strict-binding-smallm-a} and \eqref{eq:strict-binding-smallm-b}, we obtain \eqref{eq:strict-binding-smallm}, which finishes the proof.\end{proof}

\begin{remark}
From \eqref{eq:I0=frac} and \eqref{eq:ABCD}, we obtain the interesting identity
$$\lim_{m\to 0^+} \frac{I_0(m)}{m^{5/3}} = \inf_{\int |u|^2 =1} -\frac{C_u^2}{4A_u} = -\frac{c_{\rm D}^2}{4c_{\rm W}}\sup_{\int |u|^2=1} \frac{\left(\int |u|^{8/3} \right)^2}{\int |\nabla u|^2} \in (-\infty, 0).
$$
The fact that $I_0(m)/m\to 0$ as $m\to 0^+$ has been proved by Lions \cite{Lions-87}.  
\end{remark}
\end{appendix}

\end{document}